\documentclass[aps,11pt,twoside,notitlepage,superscriptaddress]{revtex4-2}

\usepackage{graphicx,epic,eepic,epsfig,amsmath,amsfonts,latexsym,amssymb,enumerate}
\usepackage{hhline}
\usepackage[mathscr]{eucal}
\usepackage[dvipsnames]{xcolor}
\usepackage{booktabs,multirow,tabularx,threeparttable,colortbl,arydshln,array}
\usepackage[colorlinks,
            linkcolor=blue,
            anchorcolor=blue,
            citecolor=red,
            ]{hyperref}

\def\squareforqed{\hbox{\rlap{$\sqcap$}$\sqcup$}}
\def\qed{\ifmmode\squareforqed\else{\unskip\nobreak\hfil
\penalty50\hskip1em\null\nobreak\hfil\squareforqed
\parfillskip=0pt\finalhyphendemerits=0\endgraf}\fi}
\def\endenv{\ifmmode\;\else{\unskip\nobreak\hfil
\penalty50\hskip1em\null\nobreak\hfil\;
\parfillskip=0pt\finalhyphendemerits=0\endgraf}\fi}

\newtheorem{theorem}{Theorem}

\newtheorem{corollary}[theorem]{Corollary}

\newtheorem{definition}[theorem]{Definition}

\newtheorem{lemma}[theorem]{Lemma}

\newtheorem{proposition}[theorem]{Proposition}
\newtheorem{remark}[theorem]{Remark}

\newenvironment{proof}[1][Proof]{\noindent\textbf{Proof.} }{\hfill\qed}

\newcommand{\nc}{\newcommand}
\nc{\rnc}{\renewcommand}
\nc{\beq}{\begin{equation}}
\nc{\eeq}{\end{equation}}
\nc{\bsp}{\begin{split}}
\nc{\esp}{\end{split}}
\nc{\beqa}{\begin{eqnarray}}
\nc{\eeqa}{\end{eqnarray}}
\nc{\lbar}[1]{\overline{#1}}
\nc{\ket}[1]{|#1\rangle}
\nc{\bra}[1]{\langle#1|}
\nc{\braket}[2]{\langle #1 | #2 \rangle}
\nc{\ketbra}[2]{|#1\rangle\!\langle#2|}
\nc{\proj}[1]{| #1\rangle\!\langle #1 |}
\nc{\avg}[1]{\langle#1\rangle}
\nc{\Rank}{\operatorname{Rank}}
\nc{\smfrac}[2]{\mbox{$\frac{#1}{#2}$}}
\nc{\tr}{\operatorname{Tr}}
\nc{\ox}{\otimes}
\nc{\dg}{\dagger}
\nc{\dn}{\downarrow}
\nc{\supp}{{\operatorname{supp}}}
\nc{\qsupp}{{\operatorname{qsupp}}}
\nc{\var}{\operatorname{var}}
\nc{\rar}{\rightarrow}
\nc{\lrar}{\longrightarrow}
\nc{\polylog}{\operatorname{polylog}}
\nc{\1}{{\openone}}
\nc{\id}{{\operatorname{id}}}
\nc{\<}{\langle}
\nc{\Hom}[2]{\mbox{Hom}(\CC^{#1},\CC^{#2})}
\nc{\rU}{\mbox{U}}
\nc{\mc}{\mathcal}
\nc{\mscr}{\mathscr}
\nc{\M}{\mathsf{M}}
\nc{\X}{\mathsf{X}}
\nc{\Y}{\mathsf{Y}}
\nc{\A}{\mathsf{A}}
\nc{\B}{\mathsf{B}}
\nc{\C}{\mathsf{C}}
\nc{\D}{\mathsf{D}}
\nc{\E}{\mathsf{E}}
\nc{\U}{\mathsf{U}}
\nc{\V}{\mathsf{V}}
\nc{\R}{\mathsf{R}}

\begin{document}

\title{Tight any-shot quantum decoupling}

\author{Mario Berta}
\affiliation{Institute for Quantum Information, RWTH Aachen University, Germany}
\author{Hao-Chung Cheng}
\affiliation{Department of Electrical Engineering and Graduate Institute of Communication Engineering, National Taiwan University, Taiwan}
\affiliation{Department of Mathematics, National Taiwan University, Taiwan}
\affiliation{Center for Quantum Science and Engineering,  National Taiwan University, Taiwan}
\affiliation{Physics/Mathematics Division, National Center for Theoretical Sciences, Taiwan}
\affiliation{Hon Hai (Foxconn) Quantum Computing Center, Taiwan}
\author{Yongsheng Yao}
\affiliation{Institute for Quantum Information, RWTH Aachen University, Germany}

\date{\today}

\begin{abstract}
Quantum information decoupling is a fundamental primitive in quantum information theory, underlying various applications in quantum physics. We prove a novel one-shot decoupling theorem formulated in terms of quantum relative entropy distance, with the decoupling error bounded by two sandwiched R\'enyi conditional entropies. In the asymptotic i.i.d.~setting of standard information decoupling via partial trace, we show that this bound is ensemble-tight in quantum relative entropy distance and thereby yields a characterization of the associated decoupling error exponent in the low-cost-rate regime. 

Leveraging this framework, we derive several operational applications formulated in terms of purified distance: (i) single-letter expressions for the exact error exponents of quantum state merging in the low-entanglement-cost and high-entanglement-distillation-rate regimes, in terms of Petz--R\'enyi conditional entropies, and
(ii) regularized expressions for achievable error exponents of entanglement distillation and quantum channel coding in terms of Petz--R\'enyi coherent informations. We further prove that these achievable bounds are tight for maximally correlated states and generalized dephasing channels, respectively, for the high distillation-rate/coding-rate regimes.

\end{abstract}

\maketitle


\section{Introduction}

Quantum information decoupling addresses the challenge of eliminating correlations between a local system and its environment through quantum evolution. This task involves transforming a bipartite quantum state $\rho_{\A\E}$ by applying a unitary operation on system $\A$, followed by a decoupling map $\mathcal{T}_{\A \rightarrow \C}$, such that the the resulting system $\C$ becomes independent to the environment $\E$. A prominent special case of this theory is standard quantum information decoupling, when $\mathcal{T}_{\A \rightarrow \C}$ is given by the partial trace over a subsystem of $\A$. As a fundamental structural pillar in quantum information theory, decoupling provides the theoretical foundation for numerous landmark results. Its applications are widespread, ranging from quantum state merging~\cite{Li_2024,Horodecki_2006,berta2009singleshotquantumstatemerging}, quantum channel simulation~\cite{Berta_2011,Li_2025,Bennett_2014}, entanglement distillation~\cite{Devetak_2005,Buscemi_2010}, to quantum channel coding~\cite{Dupuis_2009,Dupuis_2010,dupuis2010decouplingapproachquantuminformation}.

In many physically relevant scenarios,  decoupling is considered in the absence of auxiliary resources, where the unitary operation is drawn from the Haar measure and no additional catalytic systems are available. The performance of such a scheme is quantified by a divergence measuring the residual correlations between $\C$ and $\E$. A substantial body of work has been devoted to the trace distance and purified distance criteria. In particular, one-shot upper bounds on the decoupling error under the trace distance were derived for general maps $\mathcal{T}_{\A \rightarrow \C}$, expressed in terms of smooth conditional min-entropies~\cite{Dupuis_2014}. While the analysis suffices for proving standard coding theorems, these smooth-entropy-based bounds necessarily involve non-negligible fudge terms. As a consequence, they do not yield exact non-asymptotic error exponent characterizations and are primarily tailored to first- and second-order asymptotics.  To address this issue, Cheng \emph{et al.}~\cite{cheng2024jointstatechanneldecouplingoneshot} recently strengthened the one-shot bound by expressing it in terms of sandwiched conditional R\'enyi entropies, thereby clarifying the associated achievable error exponents. However, even in this refined form, the analysis remains tied to the trace distance and does not provide a sharp characterization under quantum relative entropy (or purified distance).

Quantum relative entropy plays a central role in quantum information theory.  Beyond its operational significance, bounds formulated under relative entropy can be converted into purified-distance statements via standard entropy-fidelity inequalities, thereby enabling Uhlmann-type arguments that are indispensable across a wide range of fully quantum applications. Consequently, for tasks such as quantum state merging and quantum channel coding, relative entropy thus provides a robust performance criterion. Despite its foundational importance, quantum information decoupling under the relative entropy has remained comparatively unexplored. In particular, general one-shot bounds\,---\,analogous to the well-established results for trace distance\,---\,have been notably absent from the literature. This gap constitutes a genuine technical bottleneck: Without a sharp one-shot bound formulated directly in terms of the relative entropy, a precise characterization of exact finite-blocklength error exponents remained out of reach.

In this work, we overcome this bottleneck by leveraging a recently developed trace inequality~\cite{sharp25} to establish a general one-shot upper bound on the decoupling error measured by the quantum relative entropy. 
Our bound takes an exact exponential form without the need for smoothing and remains valid for arbitrary blocklengths. 
Informally speaking, for arbitrary decoupling channels $\mc{T}_{\A \to \C}$ with Choi state $\omega_{\A\C}$ and bipartite states $\rho_{\A\E}$, we prove that

\begin{equation}
     \mathbb{E}_{\mathbb{U}(\A)} D(\mc{T}_{\A \rightarrow \C} (U_{\A} \rho_{\A\E} U_{\A}^*)\|\omega_{\C} \otimes \rho_{\E} ) 
            \leq \inf_{0<s<1} {2} G_s 2^{-s\widetilde{H}_{1+s}(\A|\E)_\rho-s\widetilde{H}_{1+s}(\A'|\C)_\omega},
\end{equation}
where $\mathbb{E}_{\mathbb{U}(\A)}$ denotes the integration over the Haar measure over the unitary group $\mathbb{U}(\A)$, $\widetilde{H}_{\alpha}(\A|\E)_\rho$ is the sandwiched conditional R\'enyi entropy defined later in Eq.~\eqref{equ:defsnad1}, and for the constant
\begin{align}
G_s:=\sup_{r>0} \frac{\log(1+r)}{r^s} \leq
\frac{s^{s-1}(1-s)^{1-s}}{\ln2}, \quad s \in (0,1].
\end{align}
A key structural feature is that the bound is directly expressed in terms of the exponents, is fully non-asymptotic, and involves no smoothing parameters or additive fudge terms.  Consequently, in the $n$-fold product setting, the induced error exponent is valid for every blocklength $n$, and is positive whenever the sum-entropy criterion \cite{Dupuis_2014}, $\widetilde{H}_{1}(\A|\E)_\rho+\widetilde{H}_{1}(\A'|\C)_\omega>0$, is satisfied. This stricture is particularly appealing for finite-resource quantum regimes and consequently for applications in physics.

For the standard decoupling scenario, we further establish a lower bound on the one-shot relative entropy decoupling error using the pinching technique combined with an operator inequality. This lower bound is expressed through an information-spectrum quantity, allowing us to obtain an ensemble-tight converse bound that matches the achievable error exponent, provided the subsystem being removed grows at a rate below a certain \emph{critical rate}. Hence, in the most operationally significant low-cost-rate regime, we determine the exact error exponent of standard quantum decoupling under the relative entropy criterion. While it remains an open question whether this ensemble tightness can be elevated to a fully general converse independent of the Haar ensemble, the appearance of a critical rate is fundamentally analogous to phenomena observed in classical data compression with quantum side information~\cite{Cheng_2021,Renes_2023}, classical-quantum channel coding~\cite{Dalai_2013,Renes_2025,Li_2025cq,CL25}, privacy amplification~\cite{Hayashi_2013, Li_2023} and catalytic decoupling~\cite{Li_2024}.

Our general framework applies in particular to decoupling maps implemented by partial isometries. In this setting, we derive an explicit error exponent for quantum state merging~\cite{Horodecki_2005,Horodecki_2006,berta2009singleshotquantumstatemerging}, when the entanglement cost rate is not too high. Furthermore, a similar approach also provides an achievability bound on the error exponents for entanglement distillation assisted by local operation and classical communication (LOCC) and a broad class of quantum communication tasks, including subspace transmission, entanglement transmission, entanglement generation as well as their one-way LOCC-assisted counterparts. Our bounds offer strong performance guarantees for both unassisted and LOCC-assisted protocols: Specifically, for coding rates below the first-order asymptotic capacity, we show that the error decays exponentially for every blocklength $n$. This provides a stronger large-deviation characterization than that offered by conventional first-order analyses. Notably, while broader axiomatic classes of operations exist, our focus remains on the physically motivated unassisted and LOCC-assisted scenarios. For certain classes of bipartite states and quantum channels, we establish corresponding exact error exponents. These applications demonstrate both the versatility of our framework and the analytical advantages of a relative entropy based approach. A comprehensive summary of the obtained error exponents is provided in Table~\ref{table:main}.\\

\begin{table}[t!]
	\centering
	\resizebox{1\columnwidth}{!}{
		\begin{tabular}{l  c c}

        \toprule
		\addlinespace

        \multirow{2}{*}{\textbf{Task}} & \textbf{Achievable Error Exponent} & \textbf{Asymptotic Tightness}

        \\

        & (for any $n$-shot) & (as $n\to \infty)$

        \\
        \addlinespace[0.2em]
        
        \midrule

        \addlinespace[0.3em]
        \multirow{2}{*}{General decoupling}
        & 
        \multirow{2}{*}{$\sup\limits_{s\in(0,1]} s\widetilde{H}_{1+s}(\A|\E)_\rho + s\widetilde{H}_{1+s}(\A'|\C)_\omega$}
        &
        \multirow{2}{*}{?}

        \\

        & &

        \\
        \addlinespace[0.3em]

        \hhline{>{\arrayrulecolor{black!25}}--->{\arrayrulecolor{black}}}

        \addlinespace[0.3em]
        \multirow{2}{*}{Standard decoupling}
        & 
        \multirow{2}{*}{$\sup\limits_{s\in(0,1]} s\left( 2r - \log |\A| + \widetilde{H}_{1+s}(\A|\E)_\rho\right)$}
        & \multirow{2}{*}{Ensemble tight}

        \\

        & &

        \\
        \addlinespace[0.3em]

        \hhline{>{\arrayrulecolor{black!25}}--->{\arrayrulecolor{black}}}

        \addlinespace[0.3em]
        Entanglement cost of
        & 
        \multirow{2}{*}{$\sup\limits_{\alpha \in[1/2,1)} \frac{1-\alpha}{2\alpha}\left(r-H^{\uparrow}_{\alpha}(\A|\B)_\rho\right)$}
        & \multirow{2}{*}{Tight}

        \\  

        quantum state merging& &
        
        \\
        \addlinespace[0.3em]
        
        \hhline{>{\arrayrulecolor{black!25}}--->{\arrayrulecolor{black}}}
        \addlinespace
         Entanglement distillation of
        & 
        \multirow{2}{*}{$\sup\limits_{s\in(0,1]} \frac{s}{2}\left(\widetilde{H}_{1+s}(\A|\R)_\rho - r\right)$}
        & \multirow{2}{*}{Tight}

        \\  

        quantum state merging& &

        \\
        \addlinespace[0.3em]
        
        \hhline{>{\arrayrulecolor{black!25}}--->{\arrayrulecolor{black}}}

        \addlinespace[0.3em]
        \multirow{3}{*}{Entanglement distillation}
        & 
        \multirow{3}{*}{$\sup\limits_{\alpha \in[1/2,1)} \!\!\frac{1-\alpha}{2\alpha}\!\left( \tfrac{1}{n}\sup\limits_{\mscr{L}_{\A^n:\B^n \rightarrow \C:\D}}I_{\alpha}(\C \rangle \D)_{\mscr{L}(\rho_{\A\B}^{\ox n})}\!-r\right)$}
        & Single-letter tight

        \\

        & & for maximally

        \\

        & & correlated states

        \\
        \addlinespace[0.3em]

        \hhline{>{\arrayrulecolor{black!25}}--->{\arrayrulecolor{black}}}

        \addlinespace[0.3em]
        \multirow{3}{*}{Quantum communication}
        & 
        \multirow{3}{*}{$ \sup\limits_{\alpha \in[1/2,1)} \!\!\tfrac{1-\alpha}{2\alpha}\!\left( \tfrac{1}{n} \sup\limits_{\phi_{\bar{\A}^n\A^n}}I_{\alpha}(\bar{\A}^n \rangle \B^n)_{\mscr{N}^{\otimes n}(\phi_{\bar{\A}^n\A^n}\!)}\!-r\right)$}
        & Single-letter tight 

        \\

        & & for generalized

        \\

        & & dephasing channels

        \\
        \addlinespace[0.1em]

        \bottomrule
        
        \end{tabular}
	}
	\caption{
		\small 
        Summary of the any $n$-shot achievable error exponents for various fully quantum tasks; see Theorems~\ref{thm:maindec}, \ref{thm:main}, \ref{thm:statemer}, \ref{thm:entan}, and \ref{thm:loccachi}. The relevant R\'enyi entropic information measures are defined in Section~\ref{sec:information-measures}. Note the the tightness in the third column is up to a certain \emph{critical rate}. The operation $\mscr{L}_{\A^n:\B^n \rightarrow \C:\D}$ corresponding to entanglement distillation protocol pertains to one-way LOCC, two-way LOCC.
	}	\label{table:main}	
\end{table}

\emph{Relation to previous work.} The closest related result is due to Cheng \emph{et al.}~\cite{cheng2024jointstatechanneldecouplingoneshot} who established a one-shot upper bound on the  decoupling error under the trace distance for arbitrary decoupling channels  $\mc{T}_{\A \rightarrow \C}$ as
\begin{equation}\label{eq:previous}
\frac{1}{2}\mathbb{E}_{\mathbb{U}(\A)} \|\mc{T}_{\A \rightarrow \C} (U_{\A} \rho_{\A\E} U_{\A}^*)-\omega_{\C} \otimes \rho_{\E} \|_1 \leq\inf_{\alpha \in [1,2]}2^{\frac{1-\alpha}{\alpha}\left(\widetilde{H}_\alpha^{\uparrow}(\A|\E)_\rho+\widetilde{H}_\alpha^{\uparrow}(\A'|\C)_\omega+\log 3\right)}.
\end{equation}
Such one-shot bounds imply achievability results for the error exponent under the trace-distance. However, the absence of corresponding converse bounds prevents a full characterization of the error exponent. In contrast, focusing on the quantum relative entropy\,---\,a more demanding criterion for measuring errors\,---\,our work establishes both achievability and ensemble-tight converse bounds for standard decoupling and shows that their coincidence in the low-rate regime. This yields the first exact characterization of the error exponent of quantum information decoupling under the quantum relative entropy. We also remark that by reverse Pinsker's inequality, the error exponent established in Eq.~\eqref{eq:previous} does not cover ours (see Remark~\ref{remark:dupuis}).

A recent work \cite{HAP24} also studies the achievability of quantum information decoupling under quantum relative entropy, where the error bound is only obtained for $\alpha = 2$. However, no tight exponent characterization was derived there, and achieving the general R\'enyi parameter regime $\alpha \in (1,2)$ remained open. Another line of work studies catalytic decoupling, where auxiliary systems assist the protocol. Tight one-shot characterizations, second-order asymptotics and large deviation analysis have been obtained under purified distance~\cite{Majenz_2017, Anshu_2017, li2024operationalinterpretationsandwichedrenyi,Li_2024}. However, these results rely on catalytic assistance and do not address the structureless setting considered here.

While recent work~\cite{lin2026exponentialanalysisentanglementdistillation} establishes the exact error exponent for entanglement distillation assisted by the more powerful class of non-entangling operations, we focus here on the physically motivated LOCC-assisted setting. Since the non-entangling distillation capacity can strictly exceed that achievable under LOCC~\cite{Lami_Regula_2024_non-entangling}, our achievability bound may at first sight appear more conservative than the corresponding result in~\cite{lin2026exponentialanalysisentanglementdistillation}. However, the converse analysis of~\cite{lin2026exponentialanalysisentanglementdistillation} continues to provide a valid upper bound in the LOCC regime. By combining these two viewpoints, we show that for maximally correlated states the optimal error exponent in the high-distillation-rate regime is already achievable using one-way LOCC.\\

\textit{Outline}.
The remainder of this paper is organized as follows. In Section~\ref{sec:notation}, we introduce the notations and definitions used throughout the paper. Section~\ref{sec:main} formalizes the quantum information decoupling problem and presents our main results. Applications of these results to quantum state merging, entanglement distillation and several quantum communication tasks are discussed in Sections~\ref{sec:statemerging}, \ref{sec:entangle} and \ref{sec:quantumcom}, respectively. Finally, Section~\ref{sec:conclu} concludes the paper with a discussion and several open problems.


\section{Preliminaries}
\label{sec:notation}

\subsection{Notation}

Throughout this paper, all logarithms are to base two.
For a finite-dimensional Hilbert space $\mathcal{H}$, let $\mc{L}(\mc{H})$ denote the set of all linear operators acting on $\mc{H}$ and let $\mc{P}(\mc{H})$ be the set of the  positive semi-definite  operators  on $\mc{H}$. The set of normalized and sub-normalized quantum states on $\mc{H}$ are defined respectively as
\begin{align*}
    \mathcal{S}(\mathcal{H})&=\{ \rho \in \mathcal{P}(\mc{H})~|~\tr \rho=1\},\\
    \mathcal{S}_{\leq}(\mathcal{H})&=\{ \rho \in \mathcal{P}(\mc{H})~|~\tr \rho \leq 1\}.
\end{align*}
We use $|\mc{H}|$ to denote the dimension of $\mc{H}$, and $I_{\mc{H}}$ for the identity operator on $\mc{H}$. When $\mc{H}$ is associated with a quantum system $\A$, the above notations $\mc{L}(\mc{H})$, $\mc{P}(\mc{H})$, $\mathcal{S}(\mc{H})$, $\mathcal{S}_{\leq}(\mc{H})$, $|\mc{H}|$ and $I_{\mc{H}}$ also can be written as $\mc{L}(\A)$, $\mc{P}(\A)$, $\mathcal{S}(\A)$, $\mathcal{S}_{\leq}(\mc{H})$, $|\A|$ and $I_{\A}$, respectively. 

For an operator $X \in \mc{L}(\mc{H})$, $\supp(X)$ denotes the support of $X$. 
For a positive semidefinite
operator $X$, the operators $\log X$ and $X^{-1}$ are evaluated on
$\supp(X)$. Whenever operator monotonicity or operator concavity of the
logarithm is invoked for possibly singular operators, we first add
$\epsilon I$ and then let $\epsilon\downarrow0$.

For $A, B \in \mathcal{P}(\mc{H})$, $\{A \geq B\}$ denotes the projection onto the positive part of $A-B$, namely  the subspace spanned by the eigenvectors corresponding to the non-negative eigenvalues of $A-B$,  $\{A>B\}$, $\{A \leq B\}$ and $\{A<B\}$ are defined analogously. A bipartite state $\rho_{\A\B} \in \mc{S}(\A\B)$ is called maximally correlated if there exist orthonormal bases $\{\ket{a_x}\}_x$ and $\{\ket{b_x}\}_x$ of $\mc{H}_{\A}$ and $\mc{H}_{\B}$ such that $\supp(\rho_{\A\B})$ is contained in the subspace spanned by $\{\ket{a_x}\ox \ket{b_x} \}_x$.

The purified distance between two sub-normalized quantum states $\rho, \sigma \in \mc{S}_{\leq}(\mc{H})$ is given by
\[
P(\rho,\sigma):=\sqrt{1-F^2(\rho,\sigma)},
\]
where $F(\rho,\sigma):=\|\sqrt{\rho}\sqrt{\sigma}\|_1+\sqrt{(1-\tr\rho)(1-\tr\sigma)}$ is the fidelity function. The trace distance between $\rho$ and $\sigma$ is defined as
\[
d(\rho,\sigma):=\frac{1}{2}\|\rho-\sigma\|_1.
\]

For isomorphic systems $\A \cong \A'$, $\Phi_{\A\A'}$ denotes the normalized maximally entangled state between $\A$ and $\A'$, given by
\[
\Phi_{\A\A'}=\frac{1}{|\A|} \sum^{|\A|}_{i,j =1}\ket{i}\bra{j}_{\A}\otimes \ket{i}\bra{j}_{\A'}.
\]

A quantum channel $\mscr{N}_{\A \rightarrow \B}$ is a completely positive and trace-preserving linear map from $\mc{L}(\A)$ to $\mc{L}(\B)$.  For a completely positive map $\mathcal{T}_{\A \rightarrow \C}$, its normalized Choi state is defined as 
\[
\omega_{\A'\C}=\mathcal{T}_{\A \rightarrow \C}(\Phi_{\A\A'}).
\]
It follows that for any quantum state $\rho_{\A\E}$,
\begin{equation}
\label{equ:choi}
  \mathcal{T}_{\A \rightarrow \C}(\rho_{\A\E})=|\A|^2 \bra{\Phi_{\A\A'}}  \omega_{\A'\C}\otimes \rho_{\A\E} \ket{\Phi_{\A\A'}}.
\end{equation}
Let $A$ be a self-adjoint operator with spectral projections $P_1, \ldots, P_r$. The pinching channel $\mathcal{E}_A$ associated with $A$ is defined as
\[
\mathcal{E}_A : X\mapsto\sum_{i=1}^r P_i X P_i.
\]
The pinching inequality~\cite{Hayashi_2002} states that for any $\sigma \in \mathcal{P}(\mathcal{H})$, we have
\begin{equation}
\sigma \leq v(A) \mathcal{E}_A(\sigma),
\end{equation}
where $v(A)$ is the number of different eigenvalues of $A$.

Let $G$ be a group with unitary representations $\{U_g\}_{g \in G}$ on the Hilbert space $\mc{H}_{\A}$ and 
 $\{V_g\}_{g \in G}$ on $\mc{H}_{\B}$. A quantum channel $\mscr{N}_{\A \rightarrow \B}$ is said to be covariant with respect to $G$ if
\begin{equation}
    V_g \mscr{N}(\cdot) V_g^*=\mscr{N}(U_g\cdot U_g^*), \quad \forall g \in G.
\end{equation}


\subsection{Quantum divergences}
\label{sec:information-measures}

The Umegaki relative entropy~\cite{Umegaki1954conditional}, also known as the quantum relative entropy, is one of the most fundamental information measures in quantum information theory. For $\rho, \sigma \in \mathcal{P}(\mathcal{H})$, it is defined as 
\begin{equation}
D(\rho\|\sigma):= \begin{cases}
\tr(\rho(\log\rho-\log\sigma)) & \text{ if }\supp(\rho)\subseteq\supp(\sigma), \\
+\infty                        & \text{ otherwise.}
                  \end{cases}
\end{equation}
For  $\rho_{\A\B} \in \mathcal{S}(\A\B)$, the conditional entropy and coherent information of $\rho_{\A\B}$ are given by
\begin{align}
H(\A|\B)_\rho:&= -\min_{\sigma_{\B} \in \mc{S}(\B)} D(\rho_{\A\B} \| I_{\A} \ox \sigma_{\B}) \\
&=-D(\rho_{\A\B} \|I_{\A} \ox \rho_{\B}), \\
I(\A \rangle \B)_\rho:&=\min_{\sigma_{\B} \in \mc{S}(\B)} D(\rho_{\A\B} \| I_{\A} \ox \sigma_{\B})\\
&=D(\rho_{\A\B} \|I_{\A} \ox \rho_{\B}).
\end{align}

Quantum Rényi divergences are parameterized families of divergence measures that extend the quantum relative entropy. Among the various proposals, the Petz R\'enyi divergence~\cite{petz1986quasi} and the sandwiched R\'enyi divergence~\cite{MDS+13,WWY14} play a particularly prominent role due to their favorable mathematical properties and well-established operational interpretations in quantum information theory.

\begin{definition}
\label{definition:sand}
Let $\alpha\in(0,+\infty)\setminus\{1\}$, $\rho, \sigma\in\mc{P}(\mc{H})$.
When $\alpha >1$ and $\supp(\rho)\subseteq\supp(\sigma)$ or $\alpha\in (0,1)$ and $\supp(\rho)\not\perp\supp(\sigma)$, the Petz R\'enyi divergence $D_{\alpha}(\rho\|\sigma)$ and the sandwiched R{\'e}nyi divergence $\tilde{D}_\alpha(\rho\|\sigma)$
are defined as
\begin{align}
&D_{\alpha}(\rho \| \sigma):=\frac{1}{\alpha-1} \log Q_{\alpha}(\rho \| \sigma) ,
\quad\text{with}\ \
Q_{\alpha}(\rho \| \sigma)=\tr \rho^\alpha \sigma^{1-\alpha}; \\
&\widetilde{D}_{\alpha}(\rho \| \sigma):=\frac{1}{\alpha-1} \log \tilde{Q}_{\alpha}(\rho \| \sigma),
\quad\text{with}\ \
\widetilde{Q}_{\alpha}(\rho \| \sigma)=\tr {({\sigma}^{\frac{1-\alpha}{2\alpha}} \rho {\sigma}^{\frac{1-\alpha}{2\alpha}})}^\alpha;
\end{align}
otherwise, we set $D_\alpha(\rho\|\sigma)=\widetilde{D}_{\alpha}(\rho \| \sigma)=+\infty$.
\end{definition}
When $\alpha$ goes to infinity,  $\tilde{D}_\alpha(\rho\|\sigma)$ converges to the max-relative entropy~\cite{Datta_2009}
\begin{equation}
D_{\rm{max}}(\rho\|\sigma):=\inf\{\lambda~|~\rho \leq 2^\lambda \sigma\}.
\end{equation}

Based on these R\'enyi divergences, one can define R\'enyi generalizations of conditional entropy and coherent information, which are widely used in the analysis of non-asymptotic quantum information theory. For $\alpha\in (0,\infty)\setminus \{1\}$ and $\rho_{\A\B} \in \mc{P}(\A\B)$, the Petz and sandwiched conditional R\'enyi entropies are defined as
\begin{align}
\label{equ:conren}
H^{\uparrow}_\alpha(\A|\B)_\rho:&=-\inf_{\sigma_{\B} \in \mc{S}(\B)} D_{\alpha}(\rho_{\A\B}\|I_{\A} \ox \sigma_{\B}), \\ \label{def:con2}
H_\alpha(\A|\B)_\rho:&= -D_{\alpha}(\rho_{\A\B}\|I_{\A} \ox \rho_{\B}), \\ 
\label{equ:defsnad}
\widetilde{H}^{\uparrow}_\alpha(\A|\B)_\rho:&=-\inf_{\sigma_{\B} \in \mc{S}(\B)} \widetilde{D}_{\alpha}(\rho_{\A\B}\|I_{\A} \ox \sigma_{\B}), \\
\label{equ:defsnad1}
\widetilde{H}_\alpha(\A|\B)_\rho:&= -\widetilde{D}_{\alpha}(\rho_{\A\B}\|I_{\A} \ox \rho_{\B}).
\end{align}
Similarly, the Petz and sandwiched R\'enyi coherent information are defined as
\begin{align}
\label{def:mul1}
I_\alpha(\A\rangle \B)_\rho:&=\inf_{\sigma_{\B} \in \mc{S}(\B)} D_{\alpha}(\rho_{\A\B}\|I_{\A} \ox \sigma_{\B}), \\ 
\label{equ:cohrenlast}
\widetilde{I}_\alpha(\A\rangle \B)_\rho:&=\inf_{\sigma_{\B} \in \mc{S}(\B)} \widetilde{D}_{\alpha}(\rho_{\A\B}\|I_{\A} \ox \sigma_{\B}).
\end{align}

In the following proposition, we summarize several important properties of these quantum R\'enyi divergences.
\begin{proposition}
\label{prop:mainpro}
For $\rho ,\sigma \in \mc{P}(\mc{H})$, the Petz R\'enyi divergence and the sandwiched R{\'e}nyi
divergence satisfy the following properties.
\begin{enumerate}[(i)]
  \item Monotonicity in $\sigma$~\cite{MDS+13,Mosonyi_2017}: if $\sigma' \geq \sigma$, then $D_{\alpha}(\rho \| \sigma') \leq D_{\alpha}(\rho \| \sigma)$,
      for $\alpha \in [0,2]$
      and $\widetilde{D}_{\alpha}(\rho \| \sigma') \leq \widetilde{D}_{\alpha}(\rho \| \sigma)$, for $\alpha \in [\frac{1}{2},+\infty)$;
   \item  Data processing inequality \cite{petz1986quasi, Frank_2013, Beigi_2013, MDS+13, Wilde_2014, Mosonyi_2017}: for any quantum channel $\mscr{N}$ from $\mc{L}(\mc{H})$ to $\mc{L}(\mc{H}')$, we have
      \begin{align}
      D_{\alpha}(\mscr{N}(\rho) \| \mscr{N}(\sigma)) &\leq D_{\alpha}(\rho \| \sigma), \quad \forall \alpha \in [0,2], \\
       \widetilde{D}_{\alpha}(\mscr{N}(\rho) \| \mscr{N}(\sigma)) &\leq \widetilde{D}_{\alpha}(\rho \| \sigma), \quad \forall \alpha \in [\tfrac{1}{2},\infty);
      \end{align}
 \item Duality relation~\cite{Hayashi_2016}: for pure state $\rho_{ABC}\in \mc{S}(ABC)$ and $\tau_{\A} \in \mc{P}(A)$ such that $\supp(\rho_A) \subseteq \supp(\tau_{\A})$, we have
  \begin{equation}
 \begin{split}
\inf_{\sigma_{\B} \in \mc{S}(\B)}\widetilde{D}_\alpha(\rho_{\A\B} \|\tau_{\A} \ox \sigma_{\B})&=- \inf_{\sigma_{\C} \in \mc{S}(C)}\widetilde{D}_\beta(\rho_{\A\C} \|\tau^{-1}_A \ox \sigma_{\C}),~\text{for}~\alpha \in [\tfrac{1}{2},+\infty],~\frac{1}{\alpha}+\frac{1}{\beta}=2,\\
 \inf_{\sigma_{\B} \in \mc{S}(\B)}
D_\alpha(\rho_{\A\B} \|\tau_{\A} \ox \sigma_{\B})&=-\widetilde{D}_\beta(\rho_{\A\C} \|\tau^{-1}_A \ox \rho_C),~\text{for}~\alpha \in [\tfrac{1}{2},+\infty],~\beta=\frac{1}{\alpha},\\
D_\alpha(\rho_{\A\B} \|\tau_{\A} \ox \rho_{\B})&=-D_\beta(\rho_{\A\C} \|\tau^{-1}_A \ox \rho_C),~\text{for}~\alpha \in [0,2],~\alpha+\beta=2;
 \end{split}
 \end{equation}
\item Isometry invariance~\cite{MDS+13}: Let $\rho_{\A\B} \in \mc{S}(\A\B)$ and Let $V_{\A \rightarrow \C}$ and $V_{\B \rightarrow \D}$ be isometries. Then all quantum conditional R\'enyi entropies and quantum R\'enyi coherent informations defined in Eq.~(\ref{equ:conren})--Eq.~(\ref{equ:cohrenlast}) are invariant under local isometries, i.e.,
\begin{equation}
    F(\rho_{\A\B})=F(V_{\A\rightarrow \C} \ox V_{\B \rightarrow \D} \rho_{\A\B}V^*_{\A\rightarrow \C} \ox V^*_{\B \rightarrow \D}),
\end{equation}
where $F$ denotes any of these quantities;

\item Additivity~\cite{MDS+13}:  For any states $\rho_{\A\B}\in\mathcal{S}(\A\B)$ and $\sigma_{\A'\B'}\in\mathcal{S}(\A'\B')$, all quantum conditional R\'enyi entropies and quantum R\'enyi coherent informations defined in Eq.~(\ref{equ:conren})--Eq.~(\ref{equ:cohrenlast}) are additive under tensor products, i.e.,
\begin{equation}
    F(\rho_{\A\B} \ox \sigma_{A'B'})=F(\rho_{\A\B})+F(\sigma_{\A'\B'}),
\end{equation}
where $F$ denotes any of these quantities.
\end{enumerate}
\end{proposition}


\section{Quantum information decoupling}

\label{sec:main}
\subsection{Achievability bound for quantum information decoupling}
\label{sub:achi}

Let $\rho_{\A\E} \in \mathcal{S}(\A\E)$ be a bipartite state. Quantum information decoupling aims to suppress correlations between systems $\A$ and $\E$ by applying a random unitary transformation on $\A$,
followed by a quantum channel $\mathcal{T}_{\A \rightarrow \C}$.

For a given state $\rho_{\A\E}$, we quantify the decoupling performance of $\mathcal{T}_{\A \rightarrow \C}$ by the expected quantum relative entropy between the resulting state and the ideal product state. Specifically, the decoupling error is defined as
\begin{equation}
  \epsilon_{\rm{dec}}(\mathcal{T}_{\A \rightarrow \C},\rho_{\A\E}) :=\mathbb{E}_{\mathbb{U}(\A)}  D(\mathcal{T}_{\A \rightarrow \C}(U_{\A} \rho_{\A\E} U_{\A}^*)\|\omega_{\C} \otimes \rho_{\E}),
\end{equation}
where the expectation is taken with respect to the Haar measure on the unitary group $\mathbb{U}(\A)$, and
$\omega_{\C}:=\mathcal{T}_{\A \rightarrow \C}(\frac{I_{\A}}{|\A|})$ denotes the output of the maximally mixed state on $A$ under $\mathcal{T}_{\A \rightarrow \C}$.

In this subsection, we study the behavior of $\epsilon_{\rm{dec}}(\mathcal{T}_{\A \rightarrow \C},\rho_{\A\E})$ for any completely positive and trace non-increasing map $\mathcal{T}_{\A \rightarrow \C}$. Our main result establishes a one-shot upper bound on the decoupling error.

\begin{theorem}
\label{thm:maindec}
For any bipartite state $\rho_{\A\E} \in \mathcal{S}(\A\E)$, non-zero, completely positive and trace non-increasing map $\mathcal{T}_{\A \rightarrow \C}$ and $s \in (0,1]$, the decoupling error satisfies
    \begin{equation}
    \label{equ:maindec}
        \begin{split}
       \epsilon_{\rm{dec}}(\mathcal{T}_{\A \rightarrow \C},\rho_{\A\E}) 
            \leq 2G_s 2^{-s\widetilde{H}_{1+s}(\A|\E)_\rho-s\widetilde{H}_{1+s}(\A'|\C)_\omega},
        \end{split}
    \end{equation}
    where $\omega_{\A'\C}=\mc{T}_{\A \rightarrow \C}(\Phi_{\A'\A})$ and $G_s:=\sup_{r>0} \frac{\log(1+r)}{r^s}  \leq \frac{1}{\ln 2} s^{s-1}(1-s)^{1-s}$ for $s\in(0,1]$.
    \end{theorem}
\begin{proof}
 When $|A|=1$, $\epsilon_{\rm{dec}}(\mathcal{T}_{\A \rightarrow \C},\rho_{\A\E})=0$, Eq.~(\ref{equ:maindec}) holds trivially. Hence, for the rest we assume that $|A| \geq 2$.
 Firstly, by direct calculation, we have
    \begin{equation}
    \label{equ:gendec1}
    \begin{split}
   &\epsilon_{\rm{dec}}(\mathcal{T}_{\A \rightarrow \C},\rho_{\A\E})   \\
   =&\mathbb{E}_{\mathbb{U}_{\A}} \tr \mc{T}_{\A \rightarrow \C}(U_{\A} \rho_{\A\E} U_{\A}^*) \log \mc{T}_{\A \rightarrow \C} (U_{\A} \rho_{\A\E} U_{\A}^*)-\mathbb{E}_{\mathbb{U}_{\A}}  \tr \mc{T}_{\A \rightarrow \C}(U_{\A} \rho_{\A\E} U_{\A}^*) \log (\omega_{\C}\ox \rho_{\E}) \\
   =&\mathbb{E}_{\mathbb{U}_{\A}} \tr \mc{T}_{\A \rightarrow \C}(U_{\A} \rho_{\A\E} U_{\A}^*) \log \mc{T}_{\A \rightarrow \C} (U_{\A} \rho_{\A\E} U_{\A}^*)-\tr \omega_{\C}\ox \rho_{\E}\log (\omega_{\C}\ox \rho_{\E}).
    \end{split}
    \end{equation}
Next, we use Eq.~(\ref{equ:choi}) to evaluate the first term in Eq.~(\ref{equ:gendec1}):
\begin{equation}
\label{equ:pregiv}
    \begin{split}
         &\mathbb{E}_{\mathbb{U}_{\A}} \tr \mc{T}_{\A \rightarrow \C}(U_{\A} \rho_{\A\E} U_{\A}^*) \log \mc{T}_{\A \rightarrow \C} (U_{\A} \rho_{\A\E} U_{\A}^*) \\
        =&\mathbb{E}_{\mathbb{U}_{\A}} |\A|^2\tr\bra{\Phi_{\A\A'}} U_{\A}\rho_{\A\E}U_{\A}^*\ox \omega_{\A'\C} \ket{\Phi_{\A\A'}} \log  |\A|^2\bra{\Phi_{\A\A'}} U_{\A}\rho_{\A\E}U_{\A}^*\ox \omega_{\A'\C} \ket{\Phi_{\A\A'}} \\
        =&\mathbb{E}_{\mathbb{U}_{\A}} |\A|^2\tr (U_{\A}\rho_{\A\E}U_{\A}^*\ox \omega_{\A'\C}) ( \Phi_{\A\A'}\ox \log  |\A|^2\bra{\Phi_{\A\A'}} U_{\A}\rho_{\A\E}U_{\A}^*\ox \omega_{\A'\C} \ket{\Phi_{\A\A'}} )
    \end{split}
\end{equation}
Because for any pure state $\phi_{\A} \in \mathcal{S}(\A)$ and 
$T_{\B} \in \mathcal{P}(\B)$, it holds that
\begin{equation}
\label{equ:purel}
    \phi_{\A} \otimes \log T_{\B}=\log (\phi_{\A} \otimes T_{\B}).
\end{equation}
(Here, by our convention, the operator logarithm is only acting on the relevant support.)
Eq.~(\ref{equ:pregiv}) and Eq.~(\ref{equ:purel}) imply that 

\begin{equation}
\label{equ:cyc}
    \begin{split}
     &\mathbb{E}_{\mathbb{U}_{\A}} \tr \mc{T}_{\A \rightarrow \C}(U_{\A} \rho_{\A\E} U_{\A}^*) \log \mc{T}_{\A \rightarrow \C} (U_{\A} \rho_{\A\E} U_{\A}^*) \\
     =& \mathbb{E}_{\mathbb{U}_{\A}} |\A|^2\tr (U_{\A}\rho_{\A\E}U_{\A}^*\ox \omega_{\A'\C})   \log (\Phi_{\A\A'}\ox |\A|^2\bra{\Phi_{\A\A'}} U_{\A}\rho_{\A\E}U_{\A}^*\ox \omega_{\A'\C} \ket{\Phi_{\A\A'}} ) \\
        =& \mathbb{E}_{\mathbb{U}_{\A}} |\A|^2\tr (\rho_{\A\E}\ox \omega_{\A'\C})  \log ( U_{\A}^*\Phi_{\A\A'}U_{\A}\ox |\A|^2\bra{\Phi_{\A\A'}} U_{\A}\rho_{\A\E}U_{\A}^*\ox \omega_{\A'\C} \ket{\Phi_{\A\A'}} ) \\
        =& \mathbb{E}_{\mathbb{U}_{\A}} \sum_{i=1}^{|\A|^2}\tr (\rho_{\A\E}\ox \omega_{\A'\C})   \log (U_{\A}^{i *}U_{\A}^*\Phi_{\A\A'}U_{\A}U_{\A}^{i}\ox |\A|^2\bra{\Phi_{\A\A'}} U_{\A}U_{\A}^i\rho_{\A\E}U_{\A}^{i*}U_{\A}^*\ox \omega_{\A'\C} \ket{\Phi_{\A\A'}} ) \\
         =& \mathbb{E}_{\mathbb{U}_{\A}} \tr (\rho_{\A\E}\ox \omega_{\A'\C})  \log (\sum_{i=1}^{|\A|^2}U_{\A}^{i *}U_{\A}^*\Phi_{\A\A'}U_{\A}U_{\A}^{i}\ox |\A|^2\bra{\Phi_{\A\A'}} U_{\A}U_{\A}^i\rho_{\A\E}U_{\A}^{i*}U_{\A}^*\ox \omega_{\A'\C} \ket{\Phi_{\A\A'}} ) \\
        =& \mathbb{E}_{\mathbb{U}_{\A}} \tr \left(\rho_{\A\E}\ox \omega_{\A'\C}\right)  \log \left(|\A|^2 \sum_{i=1}^{|\A|^2} \tr_{\tilde{\A}\bar{\A}} (U_{\A}^{i *}U_{\A}^*\Phi_{\A\A'}U_{\A}U_{\A}^{i}\ox \Phi_{\tilde{\A}\bar{\A}}\ox I_{\C\E})( I_{\A\A'}\ox U_{\tilde{\A}}U^i_{\tilde{\A}}\rho_{\tilde{\A}\E}U^{i *}_{\tilde{\A}}U_{\tilde{\A}}^*\ox \omega_{\bar{\A}\C})  \right) \\
        =& \mathbb{E}_{\mathbb{U}_{\A}} \tr \left(\rho_{\A\E}\ox \omega_{\A'\C}\right)   \log\left(|\A|^2 \sum_{i=1}^{|\A|^2}\tr_{\tilde{\A}\bar{\A}} (U_{\A}^{i *}U_{\A}^*\Phi_{\A\A'}U_{\A}U_{\A}^{i}\ox U^{i *}_{\tilde{\A}}U^*_{\tilde{\A}}\Phi_{\tilde{\A}\bar{\A}}U_{\tilde{\A}}U^i_{\tilde{\A}}\ox I_{\C\E})( I_{\A\A'}\ox \rho_{\tilde{\A}\E}\ox \omega_{\bar{\A}\C})  \right) ,
    \end{split}
\end{equation}
where $\{U_{\A}^i\}_{i=1}^{|\A|^2}$ denotes the set of Heisenberg–Weyl operators, the fourth line is due to the unitary invariance of the Haar measure, the fifth line is from the mutual orthogonality of $\{U_{\A}^{i *}U_{\A}^* \ket{\Phi}_{\A\A'}\}_{i=1}^{|\A|^2}$ for any $U_{\A} \in \mathbb{U}(\A)$ and  the block-diagonal functional calculus on the support, and the sixth and the last line come from the cyclicity of the trace.

The pinching inequality gives that 
\begin{equation}
\label{equ:newpin}
    \rho_{\A\E}\ox \omega_{\A'\C} \leq \sum_{i=1}^{|\A|^2}U_{\A}^{i *}U_{\A}^*\Phi_{\A\A'}U_{\A}U_{\A}^{i}\ox |\A|^2\bra{\Phi_{\A\A'}} U_{\A}U_{\A}^i\rho_{\A\E}U_{\A}^{i*}U_{\A}^*\ox \omega_{\A'\C} \ket{\Phi_{\A\A'}},  \quad \forall U_{\A} \in \mathbb{U}(\A).
\end{equation}
Eq.~(\ref{equ:newpin}) implies that 
\[
\supp(\rho_{\A\E}\ox \omega_{\A'\C}) \subseteq \supp (\sum_{i=1}^{|\A|^2}U_{\A}^{i *}U_{\A}^*\Phi_{\A\A'}U_{\A}U_{\A}^{i}\ox |\A|^2\bra{\Phi_{\A\A'}} U_{\A}U_{\A}^i\rho_{\A\E}U_{\A}^{i*}U_{\A}^*\ox \omega_{\A'\C} \ket{\Phi_{\A\A'}}), \quad \forall U_{\A} \in \mathbb{U}(\A).
\]
Therefore, we can further use Lemma~\ref{lem:inver} to upper bound Eq.~(\ref{equ:cyc}) as

\begin{equation}
\label{equ:gendec2}
    \begin{split}
        &\mathbb{E}_{\mathbb{U}_{\A}} \tr \mc{T}_{\A \rightarrow \C}(U_{\A} \rho_{\A\E} U_{\A}^*) \log \mc{T}_{\A \rightarrow \C} (U_{\A} \rho_{\A\E} U_{\A}^*) \\
        \leq &   \tr \left(\rho_{\A\E}\ox \omega_{\A'\C}\right) \log  \left( |\A|^2 \sum_{i=1}^{|\A|^2}  \mathbb{E}_{\mathbb{U}_{\A}}\tr_{\tilde{\A}\bar{\A}} (U_{\A}^{i *}U_{\A}^*\Phi_{\A\A'}U_{\A}U_{\A}^{i}\ox U^{i *}_{\tilde{\A}}U^*_{\tilde{\A}}\Phi_{\tilde{\A}\bar{\A}}U_{\tilde{\A}}U^i_{\tilde{\A}}\ox I_{\C\E})( I_{\A\A'}\ox \rho_{\tilde{\A}\E}\ox \omega_{\bar{\A}\C})  \right)  \\
        =&\tr \left(\rho_{\A\E}\ox \omega_{\A'\C}\right)   \log \left( |\A|^4  \mathbb{E}_{\mathbb{U}_{\A}}\tr_{\tilde{\A}\bar{\A}} (U_{\A}^*\Phi_{\A\A'}U_{\A}\ox U^*_{\tilde{\A}}\Phi_{\tilde{\A}\bar{\A}}U_{\tilde{\A}}\ox I_{\C\E})( I_{\A\A'}\ox \rho_{\tilde{\A}\E}\ox \omega_{\bar{\A}\C})  \right) \\
        =& \tr (\rho_{\A\E}\ox \omega_{\A'\C})   \log\Bigg(\tr_{\tilde{\A}\bar{\A}} \Big( \big(\frac{|\A|^2}{|\A|^2-1}I_{\A\A'\tilde{\A}\bar{\A}}+\frac{|\A|^2}{|\A|^2-1}F_{\A\tilde{\A}}\ox F_{A'\bar{\A}}-\frac{|\A|^2}{|\A|^3-|\A|}F_{\A\tilde{\A}}\ox I_{\A'\bar{\A}} \\
        &-\frac{|\A|^2}{|\A|^3-|\A|}I_{\A\tilde{\A}}\ox F_{\A'\bar{\A}}\big)\ox I_{\C\E}\Big)\big( I_{\A\A'}\ox \rho_{\tilde{\A}\E}\ox \omega_{\bar{\A}\C}\big)  \Bigg) \\
        =& \tr \rho_{\A\E}\ox \omega_{\A'\C}   \log (\frac{|\A|^2}{|\A|^2-1}I_{\A\A'}\ox \rho_{\E}\ox \omega_{\C}+\frac{|\A|^2}{|\A|^2-1}\rho_{\A\E}\ox \omega_{\A'\C}-\frac{|\A|^2}{|\A|^3-|\A|}\rho_{\A\E}\ox I_{\A'}\ox \omega_{\C}\\
        &-\frac{|\A|^2}{|\A|^3-|\A|}I_{\A}\ox \rho_{\E}\ox \omega_{\A'\C}) \\
        \leq &\tr \rho_{\A\E}\ox \omega_{\A'\C}   \log  \left\{ \frac{|\A|^2}{|\A|^2-1}\left(I_{\A\A'}\ox \rho_{\E}\ox \omega_{\C}+(1-\frac{2}{|\A|^2})\rho_{\A\E}\ox \omega_{\A'\C}\right) \right\} \\
        =  &\tr \rho_{\A\E}\ox \omega_{\A'\C} \log \left\{I_{\A\A'}\ox \rho_{\E}\ox \omega_{\C}+(1-\frac{2}{|\A|^2})\rho_{\A\E}\ox \omega_{\A'\C}\right\}+ \log \frac{|\A|^2}{|\A|^2-1}\tr \omega_{\C} ,
    \end{split}
\end{equation}
where the first equality follows from the unitary invariance of the Haar measure, and the second equality follows from Lemma~\ref{lem:avehar}. 
To justify the second inequality, we used the operator inequalities~\cite[Lemma~A.2]{tomamichel2013frameworknonasymptoticquantuminformation} 
$\rho_{\A\E}\ox I_{\A'}\ox\omega_{\C}
\geq \frac{1}{|\A|}\rho_{\A\E}\ox\omega_{\A'\C}$,
and
$I_{\A}\ox\rho_{\E}\ox\omega_{\A'\C}
\geq \frac{1}{|\A|}\rho_{\A\E}\ox\omega_{\A'\C}$.
Moreover, Eq.~(\ref{equ:newpin}), after averaging over $U_{\A}$, gives the required support inclusion for the operator mono. We then add $\epsilon I$ to both operators inside the logarithms, apply the operator monotonicity of the logarithm, and let $\epsilon\downarrow0$.

By Eq.~(\ref{equ:gendec1}) with Eq.~(\ref{equ:gendec2}), we have
\begin{equation}
\label{equ:new2}
    \begin{split}
         & \epsilon_{\rm{dec}}(\mathcal{T}_{\A \rightarrow \C},\rho_{\A\E})   \\
   \leq & \tr \rho_{\A\E}\ox \omega_{\A'\C} \left( \log \left\{I_{\A\A'}\ox \rho_{\E}\ox \omega_{\C}+(1-\frac{2}{|\A|^2})\rho_{\A\E}\ox \omega_{\A'\C}\right\} -\log ( I_{\A\A'}\ox \rho_{\E}\ox \omega_{\C})\right)
   \\
   &\quad+ \log \frac{|\A|^2}{|\A|^2-1}\tr \omega_{\C}.
    \end{split}
\end{equation}
For any $s \in (0,1]$, applying Lemma~\ref{lemm:sharp_one-shot} to $(1-\frac{2}{|\A|^2})\rho_{\A\E}\ox \omega_{\A'\C}$ and $I_{\A\A'}\ox \rho_{\E}\ox \omega_{\C}$ and using homogeneity of $\tilde{Q}_{1+s}$ in its first argument. yields
\begin{equation}
\label{equ:new3}
     \epsilon_{\rm{dec}}(\mathcal{T}_{\A \rightarrow \C},\rho_{\A\E}) \leq \log \frac{|\A|^2}{|\A|^2-1}\tr \omega_{\C}+(1-\frac{2}{|\A|^2})^sG_s \tilde{Q}_{1+s}(\rho_{\A\E}\ox \omega_{\A'\C}\|I_{\A\A'}\ox \rho_{\E}\ox \omega_{\C}).
\end{equation}

The data processing inequality implies, for every positive semidefinite operator $X$, $Y$ with $\supp(X) \subseteq \supp(Y)$,
\begin{equation}
\label{equ:data}
    \tilde{Q}_{1+s}(X\|Y) \geq (\tr X)^{1+s} (\tr Y)^{-s}.
\end{equation}
Applied to the two operators in Eq~(\ref{equ:new3}) and Eq.~(\ref{equ:data}) gives
\begin{equation}
\label{equ:datapp}
   \tilde{Q}_{1+s}(\rho_{\A\E}\ox \omega_{\A'\C}\|I_{\A\A'}\ox \rho_{\E}\ox \omega_{\C}) \geq \frac{\tr \omega_{\C}}{|\A|^{2s}} .
\end{equation}
Consequently, Eq.~(\ref{equ:new3}) and Eq.~(\ref{equ:datapp}) imply 
\begin{equation}
\label{equ:new5}
    \epsilon_{\rm{dec}}(\mathcal{T}_{\A \rightarrow \C},\rho_{\A\E}) \leq \left[\left(1-\frac{2}{|\A|^2}\right)^sG_s +|\A|^{2s} \log \frac{|\A|^2}{|\A|^2-1} \right]\tilde{Q}_{1+s}(\rho_{\A\E}\ox \omega_{\A'\C}\|I_{\A\A'}\ox \rho_{\E}\ox \omega_{\C})  .
\end{equation}

By the definition of $G_s = \sup_{r>0} \frac{\log(1+r)}{r^s}$, evaluated at $r=(|\A|^2-1)^{-1}$, we have 
\begin{equation}
  |\A|^{2s} \log \frac{|\A|^2}{|\A|^2-1} \leq (\frac{|\A|^2}{|\A|^2-1})^sG_s.
\end{equation}
Since $x \mapsto x^s$ is concave on $(0,\infty)$ for $s \in (0,1]$,
\begin{equation}
\label{equ:conlog}
  \left(1-\frac{2}{|\A|^2}\right)^s+  \left(\frac{|\A|^2}{|\A|^2-1}\right)^s\leq 2\left[\frac{1}{2}\left(1-\frac{2}{|\A|^2}+\frac{|\A|^2}{|\A|^2-1}\right) \right]^s \leq 2.
\end{equation}
Combing Eq.~(\ref{equ:new5}) to Eq.~(\ref{equ:conlog}) gives 
\begin{align}
    \epsilon_{\rm{dec}}(\mathcal{T}_{\A \rightarrow \C},\rho_{\A\E}) 
    \leq 2G_s \tilde{Q}_{1+s}(\rho_{\A\E}\ox \omega_{\A'\C}\|I_{\A\A'}\ox \rho_{\E}\ox \omega_{\C}) 
    =2G_s 2^{-s\widetilde{H}_{1+s}(\A|\E)_\rho-s\widetilde{H}_{1+s}(\A'|\C)_\omega}.
\end{align}

We complete the proof.
\end{proof}

As a direct consequence of Theorem~\ref{thm:maindec}, we obtain an achievability bound on the error exponent of quantum information decoupling, defined as the asymptotic exponential rate of decay of the decoupling error $\epsilon_{\rm{dec}}(\mathcal{T}^{\otimes n}_{\A \rightarrow \C},\rho_{\A\E}^{\ox n})$ in the i.i.d.~limit.

\begin{corollary}
    Let $\rho_{\A\E} \in \mathcal{S}(\A\E)$ and $\mathcal{T}_{\A \rightarrow \C}$ be a quantum channel. Then the error exponent of quantum information decoupling satisfies 
    \begin{equation}
       \liminf_{n\to\infty} \frac{-1}{n} \log \epsilon_{\rm{dec}}(\mathcal{T}^{\otimes n}_{\A \rightarrow \C}, \rho_{\A\E}^{\ox n})
        \geq \sup_{0 < s < 1} s(\widetilde{H}_{1+s}(\A|\E)_\rho+\widetilde{H}_{1+s}(\A'|\C)_\omega).
    \end{equation}
\end{corollary}

\begin{proof}
    By Theorem~\ref{thm:maindec}, for any $s \in (0,1)$ and $n \in \mathbb{N}$, we have
    \begin{equation}
    \label{equ:achi}
       \epsilon_{\rm{dec}}(\mathcal{T}^{\otimes n}_{\A \rightarrow \C}, \rho_{\A\E}^{\ox n}) \\
       \leq 2G_s 2^{-ns\widetilde{H}_{1+s}(\A|\E)_\rho-ns\widetilde{H}_{1+s}(\A'|\C)_\omega}.
    \end{equation}
Taking the logarithm, dividing by $-n$, and letting $n$ tend to infinity, we observe that the prefactor does not affect the exponential decay rate. Consequently,
\begin{equation}
\label{equ:achi2}
    \liminf_{n\to\infty} \frac{-1}{n} \log \epsilon_{\rm{dec}}(\mathcal{T}^{\otimes n}_{\A \rightarrow \C}, \rho_{\A\E}^{\ox n})
        \geq  s(\widetilde{H}_{1+s}(\A|\E)_\rho+\widetilde{H}_{1+s}(\A'|\C)_\omega).
\end{equation}
Since Eq.~(\ref{equ:achi2}) holds for $s \in (0,1)$, taking the supremum over $s$ completes the proof.
\end{proof}

    
\subsection{Converse bound for standard quantum information decoupling}

In subsection~\ref{sub:achi}, we derived a one-shot upper bound on the decoupling error for general quantum channels and obtained an achievable error exponent in the i.i.d. limit. In this subsection, we turn to the standard setting of quantum information decoupling, where the decoupling map is given by a partial trace.

Specifically, let $\mathcal{H}_{\A}=\mathcal{H}_{\A_1} \otimes \mathcal{H}_{\A_2}$ be a Hilbert space decomposition, and consider the map $\mathcal{T}_{\A \rightarrow \A_1}=\tr_{\A_2}$. Our first result establishes a one-shot lower bound on the decoupling error, which serves as a converse bound in this setting.

\begin{theorem}
\label{thm:stancon}
    Let $\mathcal{H}_{\A}$ be a Hilbert space with decomposition $\mathcal{H}_{\A}=\mathcal{H}_{\A_1} \otimes \mathcal{H}_{\A_2}$ and $\rho_{\A\E} \in \mathcal{S}(\A\E)$. Then the decoupling error associated with the partial trace $\tr_{\A_2}$ satisfies
    \begin{equation}
        \epsilon_{\rm{dec}}(\tr_{\A_2},\rho_{\A\E}) \geq \frac{1}{v(\rho_{\E})}\tr\!\left( \rho_{\A\E} - 9v(\rho_{\E})\frac{|\A_2|}{|\A_1|} I_{\A} \otimes \rho_{\E} \right)_{+},
    \end{equation}
where $v(\rho_{\E})$ denotes the number of different eigenvalues of $\rho_{\E}$.
\end{theorem}

\begin{proof}
Let the spectral decomposition of $\rho_{\E}$ be given by 
\[
\rho_{\E}=\sum_{i \in I} \lambda_iE_i,
\]
where $\{E_i\}_{i \in I}$ are the spectral projections.

From Lemma~\ref{lem:reldt}, we have
    \begin{equation}
    \label{equ:low}
        \begin{split}
           \epsilon_{\rm{dec}}(\tr_{\A_2},\rho_{\A\E}) &\geq \mathbb{E}_{\mathbb{U}_{\A}} \tr(\tr_{\A_2}(U_{\A} \rho_{\A\E} U_{\A}^\ast) - 9\frac{I_{\A_1}}{|\A_1|} \otimes \rho_{\E} )_+ \\
           & =\mathbb{E}_{\mathbb{U}_{\A}}\tr \left( \tr_{\A_2}(U_{\A} \rho_{\A\E} U_{\A}^\ast) \otimes \frac{I_{\A_2}}{|\A_2|} - 9\frac{I_{\A_1}}{|\A_1|} \otimes \frac{I_{\A_2}}{|\A_2|} \otimes \rho_{\E} \right)_{+} \\
           &= \mathbb{E}_{\mathbb{U}_{\A}}\tr \left( \frac{1}{|\A_2|^2} \sum_{k=1}^{|\A_2|^2} (U_k U_{\A} \rho_{\A\E} U_{\A}^\ast U_k^\ast) - 9\frac{I_{\A_1}}{|\A_1|} \otimes \frac{I_{\A_2}}{|\A_2|} \otimes \rho_{\E} \right)_{+} \\
           &\geq \mathbb{E}_{\mathbb{U}_{\A}}\tr \left( \frac{1}{|\A_2|^2} \sum_{k=1}^{|\A_2|^2} (U_k U_{\A} \mathcal{E}_{\rho_{\E}}(\rho_{\A\E}) U_{\A}^\ast U_k^\ast) - 9\frac{I_{\A_1}}{|\A_1|} \otimes \frac{I_{\A_2}}{|\A_2|} \otimes \rho_{\E} \right)_{+} \\
           &=\mathbb{E}_{\mathbb{U}_{\A}}\tr \left( \frac{1}{|\A_2|^2} \sum_{i \in I} \sum_{k=1}^{|\A_2|^2} (U_k U_{\A} E_i \rho_{\A\E} E_i U_{\A}^\ast U_k^\ast) - 9\sum_{i \in I} \frac{I_{\A_1}}{|\A_1|} \otimes \frac{I_{\A_2}}{|\A_2|} \otimes \lambda_i E_i \right)_{+},
        \end{split}
    \end{equation}
where $\{U_k\}_{k=1}^{|\A_2|^2}$ are the Heisenberg–Weyl operators on $\mc{H}_{\A_2}$, the inequality comes from the data processing inequality.
Leveraging Lemma~\ref{lem:con} , we can further lower bound Eq.~(\ref{equ:low}) as
\begin{equation}
    \begin{split}
  \epsilon_{\rm{dec}}(\tr_{\A_2},\rho_{\A\E}) &\geq   \mathbb{E}_{\mathbb{U}_{\A}} \sum_{i \in I}\sum_{k=1}^{|\A_2|^2} \tr \left( \frac{1}{|\A_2|^2}   U_k U_{\A} E_i \rho_{\A\E} E_i U_{\A}^\ast U_k^\ast- 9 \frac{I_{\A_1}}{|\A_1|} \otimes \frac{I_{\A_2}}{|\A_2|} \otimes \lambda_i E_i \right)_{+} \\
  &=\mathbb{E}_{\mathbb{U}_{\A}}\sum_{k=1}^{|\A_2|^2} \tr \left( \frac{1}{|\A_2|^2}  U_k U_{\A} \mathcal{E}_{\rho_{\E}}(\rho_{\A\E}) U_{\A}^\ast U_k^\ast - 9\frac{I_{\A_1}}{|\A_1|} \otimes \frac{I_{\A_2}}{|\A_2|} \otimes \rho_{\E} \right)_{+} \\
  &=\sum_{k=1}^{|\A_2|^2} \tr \left( \frac{1}{|\A_2|^2}  \mathcal{E}_{\rho_{\E}}(\rho_{\A\E}) - 9\frac{I_{\A_1}}{|\A_1|} \otimes \frac{I_{\A_2}}{|\A_2|} \otimes \rho_{\E} \right)_{+} \\
  &=\tr \left( \mathcal{E}_{\rho_{\E}}(\rho_{\A\E}) - 9\frac{|\A_2|}{|\A_1|} I_{\A} \otimes \rho_{\E} \right)_{+} \\
  &\geq  \frac{1}{v(\rho_{\E})}\tr\!\left( \rho_{\A\E} - 9v(\rho_{\E})\frac{|\A_2|}{|\A_1|} I_{\A} \otimes \rho_{\E} \right)_{+},
    \end{split}
\end{equation}
where the last inequality follows from the pinching inequality.
\end{proof}

Applying the one-shot lower bound in Theorem~\ref{thm:stancon} to $n$-fold i.i.d. states, we obtain a converse bound on the error exponent of standard quantum information decoupling.

\begin{corollary}
   For any $\rho_{\A\E} \in \mathcal{S}(\A\E)$ and rate of decoupling cost $r >0$, it holds that
   \begin{equation}
      \limsup_{n\to\infty}\frac{-1}{n} \log \epsilon_{\rm{dec}}(\tr_{{\A}_2^n},\rho^{\ox n}_{\A\E})
       \leq \sup_{s>0} s\left\{2r-\log|\A|+\widetilde{H}_{1+s}(\A|\E)_\rho\right\},
   \end{equation}
   where $|{\A}_2^n|=2^{nr}$.
\end{corollary}

\begin{proof}
   By Theorem~\ref{thm:stancon}, for any $n \in \mathbb{N}$, 
    \begin{equation}
    \label{equ:conexp}
    \begin{split}
      &\epsilon_{\rm{dec}}(\tr_{{\A}_2^n},\rho^{\ox n}_{\A\E}) \\
      \geq   &\frac{1}{v(\rho_{\E}^{\ox n})}\tr\!\left( \rho^{\ox n}_{\A\E} - 9v(\rho_{\E}^{\ox n})\frac{|{\A}_2^n|}{|A^n_1|} I_{\A}^{\ox n} \otimes \rho_{\E}^{\ox n} \right)_{+} \\
      =&\frac{1}{v(\rho_{\E}^{\ox n})}\tr\!\left( \rho^{\ox n}_{\A\E} - 9v(\rho_{\E}^{\ox n})\frac{2^{2nr}}{|\A|^n} I_{\A}^{\ox n} \otimes \rho_{\E}^{\ox n} \right)_{+}.
      \end{split}
    \end{equation}
For any $\delta>0$, there exists a $N_\delta \in \mathbb{N}$  such that for any $n\geq N_\delta$, it holds that
\[
\frac{\log v(\rho_{\E}^{\ox n})}{2n} \leq \delta.
\]
For $n \geq N_\delta$, we lower bound Eq.~\eqref{equ:conexp} as 
\begin{equation}
\label{equ:revcon}
    \begin{split}
      &\epsilon_{\rm{dec}}(\tr_{{\A}_2^n},\rho^{\ox n}_{\A\E}) \\
      \geq&\frac{1}{v(\rho_{\E}^{\ox n})}\tr\!\left( \rho^{\ox n}_{\A\E} - 9v(\rho_{\E}^{\ox n})\frac{2^{2nr}}{|\A|^n} I_{\A}^{\ox n} \otimes \rho_{\E}^{\ox n} \right)_{+} \\
      =&\frac{1}{v(\rho_{\E}^{\ox n})}\tr\!\left( \rho^{\ox n}_{\A\E} - 9\cdot2^{2n(r+\frac{\log v(\rho_{\E}^{\ox n})}{2n})}\frac{I_{\A}^{\ox n}}{|\A|^n} \otimes \rho_{\E}^{\ox n} \right)_{+} \\
      \geq &\frac{1}{v(\rho_{\E}^{\ox n})}\tr\!\left( \rho^{\ox n}_{\A\E} - 9\cdot2^{2n(r+\delta)}\frac{I_{\A}^{\ox n}}{|\A|^n} \otimes \rho_{\E}^{\ox n} \right)_{+} 
      \end{split}
\end{equation}
    
Eq.~\eqref{equ:revcon} and Lemma~\ref{lem:mons} give 
    \begin{equation}
    \label{equ:newupper}
    \begin{split}
        &\limsup_{n\to\infty} \frac{-1}{n} \log \epsilon_{\rm{dec}}(\tr_{{\A}_2^n},\rho^{\ox n}_{\A\E}) \\
       \leq &\sup_{s>0} s\left\{2r-\widetilde{D}_{1+s}\left(\rho_{\A\E}\big\|\tfrac{I_{\A}}{|\A|}\otimes \rho_{\E}\right)  \right\} \\
       =&\sup_{s>0} s\left\{2r-\log|\A|+\widetilde{H}_{1+s}(\A|\E)_\rho\right\}.
       \end{split}
    \end{equation}
\end{proof}

For standard quantum information decoupling, Theorem~\ref{thm:maindec} yields an $n$-shot upper bound on $\epsilon_{\rm{dec}}(\tr_{{\A}_2^n},\rho^{\ox n}_{\A\E})$:
     \begin{equation}
    \begin{split}
        \epsilon_{\rm{dec}}(\tr_{A^n_2},\rho^{\ox n}_{\A\E}) &\leq 2G_s 2^{-ns\widetilde{H}_{1+s}(\A|\E)_\rho-s\widetilde{H}_{1+s}({A^n_1}'{A_2^n}'|A_1^n)_\omega} \\
        &=2G_s 2^{-ns\widetilde{H}_{1+s}(\A|\E)_\rho-s(\log|A_2^n|-\log|A_1^n|)} \\
        &=2G_s 2^{-ns\widetilde{H}_{1+s}(\A|\E)_\rho+n\log|\A|-2nsr}.
        \end{split}
    \end{equation}
 This bound implies the following achievability bound on the error exponent of standard quantum information decoupling:
    \begin{equation}
    \label{equ:nuwlower}
       \liminf_{n\to\infty} \frac{-1}{n} \log \epsilon_{\rm{dec}}(\tr_{{\A}_2^n},\rho^{\ox n}_{\A\E})
       \geq \sup_{s \in (0,1)} s\left\{2r-\log|\A|+\widetilde{H}_{1+s}(\A|\E)_\rho\right\}. 
    \end{equation}
We next show that this achievability bound matches the converse bound when the rate $r$ is below a critical threshold. As a consequence, we obtain an exact characterization of the error exponent in the low-rate regime for standard quantum information decoupling.

\begin{theorem}
\label{thm:main}
Let $\rho_{\A\E} \in \mathcal{S}(\A\E)$ and $0<r\leq R_{\rm critical}
:=
\frac12
\left(
 \log|\A|
 -
 \left.
 \frac{\mathrm d}{\mathrm ds}
 \left[
  s\widetilde H_{1+s}(\A|\E)_\rho
 \right]
 \right|_{s=1}
\right)$. The error exponent of standard quantum information decoupling satisfies

\begin{equation}
\label{equ:mainexact}
\begin{split}
    \lim_{n \to \infty} -\frac{1}{n} \log \epsilon_{\rm{dec}}(\tr_{{\A}_2^n},\rho^{\otimes n}_{AE})
    =  \sup_{s \in (0,1)} \left\{ s\left( 2r-\log|\A|+\widetilde{H}_{1+s}(\A|\E)_\rho \right) \right\}.
\end{split}
\end{equation}
\end{theorem}
\begin{proof}
 Eq.~\eqref{equ:mainexact} follows from the fact that when $r \leq R_{\rm{critical}}$, the lower bound of Eq.~\eqref{equ:nuwlower} and the upper bound of Eq.~\eqref{equ:newupper} coincide. To see this, consider the function
 \[
 f(s)=s\left( 2r-\log|\A|+\widetilde{H}_{1+s}(\A|\E)_\rho \right)
 \]
for $s \in [0,\infty)$. The convexity of $s \rightarrow s\widetilde{D}_{1+s}(\rho_{\A\E} \| I_{\A}\otimes \rho_{\E})$  shown in~\cite[Lemma~III. 12]{Mosonyi_2017} implies that $f(s)$ is concave. So $\sup_{s \in [0,1]} f(s)=\sup_{s\geq 0} f(s)$ holds if and only if $f'(1) \leq 0$. Since
\[
f'(1)=2\bigl(r-R_{\rm critical}\bigr),
\]
the condition $f'(1)\leq0$ is equivalent to
$r\leq R_{\rm critical}$.
\end{proof}

\begin{remark} \label{remark:dupuis}
By reverse Pinsker's inequality and substitution $\alpha = \frac{1}{1-s}$, \eqref{eq:previous} implies an achievable error exponent under the quantum relative entropy of the form
\begin{align}
    \frac{1}{2} \sup_{s \in (0,1)} \left\{ s\left( 2r-\log|\A|+\widetilde{H}^{\uparrow}_{\frac{1}{1-s}}(\A|\E)_\rho \right) \right\},
\end{align}
which is weaker than \eqref{equ:nuwlower} because $\widetilde{H}_{1+s}(\A|\E)_{\rho} \geq \widetilde{H}^{\uparrow}_{\frac{1}{1-s}}(\A|\E)_{\rho}$ for all $s\in(0,1]$ \cite[Eq.~(48)]{TBH14}.
\end{remark}


\section{Quantum state merging}
\label{sec:statemerging}

Let $\rho_{\A\B\R}$ be a tripartite pure state. Alice, Bob and a referee hold system $\A$, $\B$ and $\R$, respectively. Quantum state merging is the task of transmitting the quantum information contained in system $\A$ from Alice to Bob, such that Bob eventually holds systems $\A$ and $\B$, while preserving the global state with the reference system $\R$.

A quantum state merging protocol consists of a shared maximally entangled state $\Phi_{\bar{\A}\bar{\B}}$ between Alice and Bob, together with an LOCC channel $\Lambda_{\A\bar{\A}:\B\bar{\B} \rightarrow \tilde{\A}:\A\B\tilde{\B}}$, which transforms the joint state $\rho_{\A\B\R} \otimes \Phi_{\bar{\A}\bar{\B}}$ into a state close to $\rho_{\A\B\R} \otimes \Phi_{\tilde{\A}\tilde{\B}}$, where $\Phi_{\tilde{\A}\tilde{\B}}$ is a maximally entangled state. The performance of the protocol is quantified by the purified distance between the final state and the ideal target state $\rho_{\A\B\R} \otimes \Phi_{\tilde{\A}\tilde{\B}}$.

It is well known that in order to achieve asymptotically perfect quantum state merging, entanglement must be consumed when $H(\A|\R)_\rho<0$, whereas entanglement can be distilled when $H(\A|\R)_\rho>0$. Accordingly, in the former case we are concerned with the entanglement cost $\log|\bar{\A}|-\log|\tilde{\A}|$, while in the latter case we consider the entanglement distillation $\log|\tilde{\A}|-\log|\bar{\A}|$.

When $H(\A|\R)_\rho<0$, we define the optimal error among all quantum state merging protocols with a fixed \emph{entanglement cost} $k$ as
\[
P^{\rm{merg},c}_{\A \rightarrow \B}(\rho_{\A\B\R},k):=\min_{\mathcal{M}} P(\mathcal{M}(\rho_{\A\B\R}),\rho_{\A\B\R}\otimes \psi_{\tilde{\A}\tilde{\B}}),
\]
where the minimization is taken over all quantum state merging protocols $\mathcal{M}$ with entanglement cost $k$. When $H(\A|\R)_\rho>0$, the optimal error $P^{\rm{merg},d}_{\A \rightarrow \B}(\rho_{\A\B\R},k)$ for a fixed \emph{entanglement distillation} $k$ is defined analogously.

For a fixed entanglement cost~(distillation) rate 
$r$, he error exponent of quantum state merging is defined as the exponential decay rate of $P^{\rm{merg},c}_{\A^n \rightarrow \B^n}(\rho_{\A\B\R}^{\otimes n},nr)$~($P^{\rm{merg},d}_{\A^n \rightarrow \B^n}(\rho_{\A\B\R}^{\otimes n},nr)$) as $n \rightarrow \infty$, i.e.,
\begin{align}
    &E^{\rm{merg},c}(\rho_{\A\B\R},r):=\liminf_{n \rightarrow \infty} \frac{-1}{n} \log P^{\rm{merg},c}_{\A^n \rightarrow \B^n}(\rho_{\A\B\R}^{\otimes n},nr), \\
    &E^{\rm{merg},d}(\rho_{\A\B\R},r):=\liminf_{n \rightarrow \infty} \frac{-1}{n} \log P^{\rm{merg},d}_{\A^n \rightarrow \B^n}(\rho_{\A\B\R}^{\otimes n},nr). 
\end{align}

In this subsection, we study the error exponent of quantum state merging. We first address the achievability part. The main technical tool is the following application of Theorem~\ref{thm:maindec} to partial isometries.

\begin{lemma}
    \label{lem:main}
Let $\rho_{\A\R} \in \mathcal{S}(\A\R)$, $\mathcal{H}_{\bar{\A}} \subseteq \mathcal{H}_{\A}$ be a subspace and $V_{\bar{\A} \rightarrow \tilde{\A}}$ a partial isometry from $\mathcal{H}_{\bar{\A}}$ onto $\mathcal{H}_{\tilde{\A}}$. Then, for any $s \in (0,1)$, we have
\begin{equation}
    \mathbb{E}_{\mathbb{U}(\A)} D\left(VU(\rho_{\A\R})U^*V^* \| \frac{I_{\tilde{\A}}}{|\A|}\otimes \rho_{\R}\right) \leq 2G_s \frac{|\tilde{\A}|^{1+s}}{|\A|} 2^{-s\widetilde{H}_{1+s}(\A|\R)_\rho}.
\end{equation}
\end{lemma}
\begin{proof}
    Applying Theorem~\ref{thm:maindec}, we obtain
    \begin{equation}
    \begin{split}
        &\mathbb{E}_{\mathbb{U}(\A)} D\left(VU(\rho_{\A\R})U^*V^* \| \frac{I_{\tilde{A}}}{|\A|}\otimes \rho_{\R}\right) \\
        \leq &2G_s 2 ^{-s\widetilde{H}_{1+s}(\A|\R)_\rho-s\widetilde{H}_{1+s}(\A'|\tilde{\A})_{V\Phi_{\A'\A}V^*}}\\
        =&2G_s \frac{|\tilde{\A}|^{1+s}}{|\A|} 2^{-s\widetilde{H}_{1+s}(\A|\R)_\rho},
        \end{split}
    \end{equation}
    where the equality follows from $-s\widetilde{H}_{1+s}(\A'|\tilde{\A})_{V\Phi_{\A'\A}V^*}=\log\frac{|\tilde{\A}|^{1+s}}{|\A|}$.
\end{proof}

With Lemma~\ref{lem:main} in hand, we can establish an achievability bound on the error exponent for quantum state merging stated as follows.

\begin{theorem}
    \label{thm:statemer}
   Let $\rho_{\A\B\R}$ be a tripartite pure state. If $H(\A|\R)_\rho>0$ and $0<r<H(\A|\R)_\rho$, then for any $0<s<1$ and $n \in \mathbb{N}$, we have
  \begin{equation}
  \label{equ:main}
 P^{\rm{merg},d}_{A^n \rightarrow B^n}(\rho^{\otimes n}_{\A\B\R},nr)   
 \leq\sqrt{2G_s 2^{nrs}2^{-sn\widetilde{H}_{1+s}(\A|\R)_{\rho}} }.
 \end{equation}
Consequently, an achievability bound on the error exponent follows: 
 \begin{equation}
     E^{\rm{merg},d}(\rho_{\A\B\R},r)   \geq \frac{1}{2}\sup_{0<s<1} s\left(\widetilde{H}_{1+s}(\A|\R)_\rho-r\right)
     =\frac{1}{2}\sup_{0<s<1} s\left(-H^{\uparrow}_{\frac{1}{1+s}}(\A|\B)_\rho-r\right).
 \end{equation}
 If $H(\A|\R)_\rho<0$ and $r>-H(\A|\R)_\rho$, then
  \begin{equation}
     E^{\rm{merg},c}(\rho_{\A\B\R},r)   \geq \frac{1}{2}\sup_{0<s<1} s\left(\widetilde{H}_{1+s}(\A|\R)_\rho+r\right)
     =\frac{1}{2}\sup_{0<s<1} s\left(-H^{\uparrow}_{\frac{1}{1+s}}(\A|\B)_\rho+r\right).
 \end{equation}
\end{theorem}

\begin{proof}
We consider the two cases separately.

\textbf{Case 1:} $H(\A|\R)_\rho>0$;

For simplicity, we assume that $2^{nr}$ is an integer. Otherwise, it
can be replaced by $\lfloor2^{nr}\rfloor$, which does not affect the
error exponent.
Let $A_{n,r}$ be the smallest integer not smaller than $|\A|^n$ that
is divisible by $2^{nr}$, $\mathcal{H}_{\bar{\A}_n}$, $\mathcal{H}_{\tilde{\A}_n}$ be the Hilbert space of dimension $|\bar{\A}_n|=\A_{n,r}$, $|\tilde{\A}_n|=2^{nr}$, respectively and $V_{\A^n \rightarrow \bar{\A}_n}$ be an isometry from $\mathcal{H}_{\A}^{\otimes n}$ to $\mathcal{H}_{\bar{\A}_n}$.

We decompose $\mathcal{H}_{\bar{\A}_n}$ into a direct sum of mutually orthogonal subspaces $\{\mathcal{H}_{\A_i}\}_{i=1}^{\frac{\A_{n,r}}{2^{nr}}}$, each of dimension $2^{nr}$. For each $i$, let $V_i$ be a partial isometry from $\mathcal{H}_{\A_i}$ to $\mathcal{H}_{\tilde{\A}_n}$.

Applying  Lemma~\ref{lem:main} to each $V_i$, we obtain
\begin{equation}
\label{equ:upper}
\begin{split}
   &\mathbb{E}_{\mathbb{U}(\bar{\A}_n)} \sum_{i=1}^{\frac{A_{n,r}}{2^{nr}}} D\left(V_iU V_{\A^n \rightarrow \bar{\A}_n}(\rho_{\A\R}^{\otimes n})V^*_{\A^n \rightarrow \bar{\A}_n}U^*V_i^*\| \frac{I_{\tilde{A}_n}}{A_{n,r}}\otimes \rho_{\R}^{\otimes n}\right) \\
   \leq &\sum_{i=1} ^{\frac{A_{n,r}}{2^{nr}}} 2G_s \frac{2^{nr(1+s)}}{A_{n,r}} 2^{-s\widetilde{H}_{1+s}(\bar{\A}_n|R^n)_{V_{\A^n \rightarrow \bar{\A}_n}\rho^{\otimes n}V^*_{\A^n \rightarrow \bar{\A}_n}}} \\
   =&\sum_{i=1} ^{\frac{A_{n,r}}{2^{nr}}} 2G_s \frac{2^{nr(1+s)}}{A_{n,r}} 2^{-sn\widetilde{H}_{1+s}(\A|\R)_{\rho}} \\
   =& 2G_s 2^{nrs}2^{-sn\widetilde{H}_{1+s}(\A|\R)_{\rho}} ,
   \end{split}
\end{equation}
where the first equality comes from Proposition~(\romannumeral4) and (\romannumeral5).

For each $i$, define
\[
p_i=\tr V_i UV_{\A^n \rightarrow \bar{\A}_n}(\rho_{\A\R}^{\otimes n})V^*_{\A^n \rightarrow \bar{\A}_n}U^*V_i^*.
\]
Using the inequality~$D(\rho\|\sigma) \geq P^2(\rho,\sigma)$, we lower bound the first term in \eqref{equ:upper} as
\begin{equation}
\label{equ:lower1}
\begin{split}
    &\mathbb{E}_{\mathbb{U}(\bar{\A}_n)} \sum_{i=1}^{\frac{A_{n,r}}{2^{nr}}} D(V_iU V_{\A^n \rightarrow \bar{\A}_n}(\rho_{\A\R}^{\otimes n})V^*_{\A^n \rightarrow \bar{\A}_n}U^*V_i^* \big\| \tfrac{I_{\tilde{\A}_n}}{A_{n,r}}\otimes \rho_{\R}^{\otimes n})  \\
=&\mathbb{E}_{\mathbb{U}(\bar{\A}_n)} \Big\{\sum_{i=1}^{\frac{A_{n,r}}{2^{nr}}} p_i D\left(\tfrac{1}{p_i}V_iU V_{\A^n \rightarrow \bar{\A}_n}(\rho_{\A\R}^{\otimes n})V^*_{\A^n \rightarrow \bar{\A}_n}U^*V_i^* \big\| \tfrac{I_{\tilde{\A}_n}}{|\tilde{\A}_n|}\otimes \rho_{\R}^{\otimes n}\right)
+ D\left(\{p_1,\ldots, p_{\frac{A_{n,r}}{2^{nr}}}\} \big \| \{\tfrac{2^{nr}}{A_{n,r}},\ldots, \tfrac{2^{nr}}{A_{n,r}}\}\right) \Big\}
\\
\geq& \mathbb{E}_{\mathbb{U}(\bar{\A}_n)} \sum_{i=1}^{\frac{A_{n,r}}{2^{nr}}} p_i D\left(\tfrac{1}{p_i}V_iU V_{\A^n \rightarrow \bar{\A}_n}(\rho_{\A\R}^{\otimes n})V^*_{\A^n \rightarrow \bar{\A}_n}U^*V_i^* \big\| \tfrac{I_{\tilde{\A}_n}}{|\tilde{\A}_n|}\otimes \rho_{\R}^{\otimes n}\right) 
\\
\geq &  \mathbb{E}_{\mathbb{U}(\bar{\A}_n)} \sum_{i=1}^{\frac{A_{n,r}}{2^{nr}}} p_i P^2\left(\tfrac{1}{p_i}{V_iU V_{\A^n \rightarrow \bar{\A}_n}(\rho_{\A\R}^{\otimes n})V^*_{\A^n \rightarrow \bar{\A}_n}U^*V_i^*},\tfrac{I_{\tilde{\A}_n}}{|\tilde{\A}_n|}\otimes \rho_{\R}^{\otimes n}\right) 
\\
= &  \mathbb{E}_{\mathbb{U}(\bar{\A}_n)} \left(1-\sum_{i=1}^{\frac{A_{n,r}}{2^{nr}}} p_i F^2\left(\tfrac{1}{p_i}V_iU V_{\A^n \rightarrow \bar{\A}_n}(\rho_{\A\R}^{\otimes n})V^*_{\A^n \rightarrow \bar{\A}_n}U^*V_i^*,\tfrac{I_{\tilde{\A}_n}}{|\tilde{\A}_n|}\otimes \rho_{\R}^{\otimes n}\right)\right).
\end{split}
\end{equation}
By Ulmman's theorem, for each $i$, there exists an isometry $H^i_{\B^n \rightarrow \tilde{\B}_n\A^n\B^n}$ such that
\begin{equation}
\label{equ:umequi}
\begin{split}
&F^2\left(H^i_{\B^n \rightarrow \tilde{\B}_n\A^n\B^n}(\tfrac{1}{p_i}{V_i UV_{\A^n \rightarrow \bar{\A}_n}(\rho_{\A\B\R}^{\otimes n})V^*_{\A^n \rightarrow \bar{\A}_n}U^*V_i^*}){H^{i*}_{\B^n \rightarrow \tilde{\B}_n\A^n\B^n}}, \Phi_{\tilde{\A}_n\tilde{\B}_n} \otimes \rho_{\A\B\R}^{\otimes n}\right) 
\\
=&F^2\left(\tfrac{1}{p_i}{V_iU V_{\A^n \rightarrow \bar{\A}_n}(\rho_{\A\R}^{\otimes n})V^*_{\A^n \rightarrow \bar{\A}_n}U^*V_i^*},\tfrac{I_{\tilde{\A}_n}}{|\tilde{\A}_n|}\otimes \rho_{\R}^{\otimes n}\right),
\end{split}
\end{equation}
where $\Phi_{\tilde{\A}_n\tilde{\B}_n}$ is the maximally entangled state on  $\tilde{\A}_n\tilde{\B}_n$. 

Combining Eq.~(\ref{equ:lower1}) and Eq.~(\ref{equ:umequi}), and using the concavity of fidelity, we have 
\begin{equation}
\label{equ:lower2}
\begin{split}
&\mathbb{E}_{\mathbb{U}(\bar{\A}_n)} \sum_{i=1}^{\frac{A_{n,r}}{2^{nr}}} D\left(V_iU V_{\A^n \rightarrow \bar{\A}_n}(\rho_{\A\R}^{\otimes n})V^*_{\A^n \rightarrow \bar{\A}_n}U^*V_i^* \big\| \tfrac{I_{\tilde{\A}_n}}{A_{n,r}}\otimes \rho_{\R}^{\otimes n}\right)  
\\
\geq &\mathbb{E}_{U \sim \rm{Haar}} \left(1-\sum_{i=1}^{\frac{A_{n,r}}{2^{nr}}} p_i F^2\left(H^i_{\B^n \rightarrow \tilde{\B}_n\A^n\B^n}(\tfrac{1}{p_i}{V_i UV_{\A^n \rightarrow \bar{\A}_n}(\rho_{\A\B\R}^{\otimes n})V^*_{\A^n \rightarrow \bar{\A}_n}U^*V_i^*}){H^{i *}_{\B^n \rightarrow \tilde{\B}_n\A^n\B^n}}, \Phi_{\tilde{\A}_n\tilde{\B}_n} \otimes \rho_{\A\B\R}^{\otimes n}\right)\right) \\
\geq & \mathbb{E}_{U \sim \rm{Haar}} \left(1- F^2(\sum_{i=1}^{\frac{A_{n,r}}{2^{nr}}} H^i_{\B^n \rightarrow \tilde{\B}_n\A^n\B^n}(V_i UV_{\A^n \rightarrow \bar{\A}_n}(\rho_{\A\B\R}^{\otimes n})V^*_{\A^n \rightarrow \bar{\A}_n}U^*V_i^*){H^{i *}_{\B^n \rightarrow \tilde{\B}_n\A^n\B^n}}, \Phi_{\tilde{\A}_n\tilde{\B}_n} \otimes \rho_{\A\B\R}^{\otimes n})\right),
\end{split}
\end{equation}
Eq.~(\ref{equ:lower2}) and Eq.~(\ref{equ:upper}) imply that there exists a unitary operator $U$ such that
\begin{equation}
\label{equ:newstmer}
    \begin{split}
     &P^{\rm{merg},d}_{\A^n \rightarrow \B^n}(\rho_{\A\B\R}^{\otimes n},nr)   \\
 \leq &   P\left(\sum_{i=1}^{\frac{A_{n,r}}{2^{nr}}} H^i_{\B^n \rightarrow \tilde{\B}_n\A^n\B^n}(V_i UV_{\A^n \rightarrow \bar{\A}_n}(\rho_{\A\B\R}^{\otimes n})V^*_{\A^n \rightarrow \bar{\A}_n}U^*V_i^*)H^{i *}_{\B^n \rightarrow \tilde{\B}_n\A^n\B^n}, \Phi_{\tilde{\A}_n\tilde{\B}_n} \otimes \rho_{\A\B\R}^{\otimes n}\right) \\
 \leq & \sqrt{2G_s 2^{nrs}2^{-sn\widetilde{H}_{1+s}(\A|\R)_{\rho}} }.
    \end{split}
\end{equation}

From Eq.~(\ref{equ:newstmer}), we have

\begin{equation}
     E^{\rm{merg},d}(\rho_{\A\B\R},r)  \geq \frac{1}{2}\sup_{0<s<1} s\left(\tilde{H}_{1+s}(\A|\R)_\rho-r\right)=\frac{1}{2}\sup_{0<s<1} s\left(-H^{\uparrow}_{\frac{1}{1+s}}(\A|\B)_\rho-r\right),
 \end{equation}
 where the equality is due to Proposition~(\romannumeral3).

\textbf{Case 2:} $H(\A|\R)<0$.

We choose $r'>r$ and let  $\Phi_{\bar{\A}\bar{\B}}$ be a maximally entangled state with $|\bar{\A}|=2^{r'}$.
Then $H(\A\bar{\A}|\R)_{\rho \otimes \Phi}>0$ and \textbf{Case 1} applies to $\rho_{\A\B\R}^{\otimes n}\otimes \Phi_{\bar{\A}\bar{\B}}^{\otimes n} $. Consequently, there exists a LOCC channel $\Lambda_{\A^n\bar{\A}^n:\B^n\bar{\B}^n \rightarrow \tilde{\A}_n:\A^nB^n\bar{\A}^n\bar{\B}^n\tilde{\B}_n}$ such that 
\begin{equation}
\label{equ:larmer}
\begin{split}
    &P\left(\Lambda_{\A^n\bar{\A}^n:\B^n\bar{\B}^n \rightarrow \tilde{\A}_n:\A^nB^n\bar{\A}^n\bar{\B}^n\tilde{\B}_n}(\rho_{\A\B\R}^{\otimes n} \otimes \Phi_{\bar{\A}\bar{\B}}^{\otimes n}),\rho_{\A\B\R}^{\otimes n} \otimes \Phi_{\bar{\A}\bar{\B}}^{\otimes n} \otimes \Phi_{\tilde{\A}_n\tilde{\B}_n}\right) \\
    \leq &\sqrt{2G_s 2^{n(r'-r)s}2^{-sn\widetilde{H}_{1+s}(\A\bar{\A}|\R)_{\rho \otimes \Phi}} },
    \end{split}
\end{equation}
where $\tilde{\A}_n=\tilde{\B}_n=2^{nr'-nr}$. Eq.~(\ref{equ:larmer}) implies that if we apply the LOCC channel $$\tr_{\bar{\A}^n\bar{B}^n} \circ\Lambda_{\A^n\bar{\A}^n:\B^n\bar{\B}^n \rightarrow \tilde{\A}_n:\A^nB^n\bar{\A}^n\bar{\B}^n\tilde{\B}_n}$$ to $\rho_{\A\B\R}^{\otimes n}\otimes \Phi_{\bar{\A}\bar{\B}}^{\otimes n}$, we can obtain
\begin{equation}
    \label{equ:encost}
\begin{split}
&P^{\rm{merg},c}_{\A^n \rightarrow \B^n}(\rho_{\A\B\R}^{\otimes n},nr)   \\
\leq & P\left(\tr_{\bar{\A}^n\bar{B}^n}\circ\Lambda_{\A^n\bar{\A}^n:\B^n\bar{\B}^n \rightarrow \tilde{\A}_n:\A^nB^n\bar{\A}^n\bar{\B}^n\tilde{\B}_n}(\rho_{\A\B\R}^{\otimes n} \otimes \Phi_{\bar{\A}\bar{\B}}^{\otimes n}),\rho_{\A\B\R}^{\otimes n}  \otimes \Phi_{\tilde{\A}_n\tilde{\B}_n}\right) \\
\leq &P\left(\Lambda_{\A^n\bar{\A}^n:\B^n\bar{\B}^n \rightarrow \tilde{\A}_n:\A^nB^n\bar{\A}^n\bar{\B}^n\tilde{\B}_n}(\rho_{\A\B\R}^{\otimes n} \otimes \Phi_{\bar{\A}\bar{\B}}^{\otimes n}),\rho_{\A\B\R}^{\otimes n} \otimes \Phi_{\bar{\A}\bar{\B}}^{\otimes n} \otimes \Phi_{\tilde{\A}_n\tilde{\B}_n}\right) \\
\leq & \sqrt{2G_s 2^{n(r'-r)s}2^{-sn\widetilde{H}_{1+s}(\A\bar{\A}|\R)_{\rho \otimes \Phi}} }\\
=&\sqrt{2G_s 2^{-nrs}2^{-sn\widetilde{H}_{1+s}(\A|\R)_{\rho }} }.
\end{split}
\end{equation}
Eq.~(\ref{equ:encost}) gives 
\begin{equation}
    E^{\rm{merg},c}(\rho_{\A\B\R},r)  \geq \frac{1}{2}\sup_{0<s<1} s\left(\widetilde{H}_{1+s}(\A|\R)_\rho+r\right)=\frac{1}{2}\sup_{0<s<1} s\left(-H^{\uparrow}_{\frac{1}{1+s}}(\A|\B)_\rho+r\right).
 \end{equation}
\end{proof}

In the following, we derive a converse bound on the error exponent for quantum state merging. Our approach relies on the one-shot characterization of the minimal entanglement cost $E(\rho_{\A\B\R},\epsilon)$  for merging the state $\rho_{\A\B\R}$ within error $\epsilon$.

It has been shown~ that  $E(\rho_{\A\B\R},\epsilon)$  can be lower bounded by the partially smoothed conditional-min entropy~\cite{Anshu_2020}:
\begin{equation}
\label{equ:con2}
    E(\rho_{\A\B\R},\epsilon) \geq -H^\epsilon_{\rm{min}}(\A|\dot{\R})_\rho\geq  -H^\epsilon_{\rm{min}}(\A|\R)_\rho,
\end{equation}
where
\begin{align*}
H^\epsilon_{\rm min}(\A|\dot{\R})_\rho
&:=
-\inf_{\substack{
 \tilde\rho_{\A\R}\in\mc S_{\leq}(\A\R):\\
 P(\tilde\rho_{\A\R},\rho_{\A\R})\leq\epsilon,\;
 \tilde\rho_{\R}\leq\rho_{\R}
}}
D_{\rm max}
\left(
 \tilde\rho_{\A\R}
 \middle\|
 I_{\A}\ox\rho_{\R}
\right),\\
H^\epsilon_{\rm min}(\A|\R)_\rho
&:=
-\inf_{\substack{
 \tilde\rho_{\A\R}\in\mc S_{\leq}(\A\R):\\
 P(\tilde\rho_{\A\R},\rho_{\A\R})\leq\epsilon
}}
D_{\rm max}
\left(
 \tilde\rho_{\A\R}
 \middle\|
 I_{\A}\ox\rho_{\R}
\right).
\end{align*}
are the partially smoothed conditional-min entropy and smooth conditional-min entropy.

Equivalently, the maximal entanglement distillation for merging the state $\rho_{\A\B\R}$ within error $\epsilon$, $D(\rho_{\A\B\R},\epsilon)=-E(\rho_{\A\B\R},\epsilon) $, satisfies
\begin{equation}
\label{equ:con1}
    D(\rho_{\A\B\R},\epsilon) \leq H^\epsilon_{\rm{min}}(\A|\dot{\R})_\rho\leq  H^\epsilon_{\rm{min}}(\A|\R)_\rho.
\end{equation}
It follows that $P^{\rm{merg},c}_{\A \rightarrow \B}(\rho_{\A\B\R},\lambda)$ and $P^{\rm{merg},d}_{\A \rightarrow \B}(\rho_{\A\B\R},\lambda)$ satisfy
\begin{equation}
     \begin{split}
         &P^{\rm{merg},c}_{\A \rightarrow \B}(\rho_{\A\B\R},\lambda)\\
         =&\min\{\delta~|~   E(\rho_{\A\B\R},\delta)\leq \lambda \} \\
         \geq &\min\{\delta~|~   H^\delta_{\rm{min}}(\A|\R)_\rho\geq -\lambda   \} \\
         =&\min\{P(\tilde{\rho}_{\A\R},\rho_{\A\R})~|~\tilde{\rho}_{\A\R}\leq 2^{\lambda}I_{\A} \ox \rho_{\R}  \}.
     \end{split}
 \end{equation}
and
 \begin{equation}
     \begin{split}
         &P^{\rm{merg,d}}_{\A \rightarrow \B}(\rho_{\A\B\R},\lambda)\\
         =&\min\{\delta~|~   D(\rho_{\A\B\R},\delta)\geq \lambda \} \\
         \geq &\min\{\delta~|~   H^\delta_{\rm{min}}(\A|\R)_\rho\geq \lambda   \} \\
         =&\min\{P(\tilde{\rho}_{\A\R},\rho_{\A\R})~|~\tilde{\rho}_{\A\R}\leq 2^{-\lambda}I_{\A} \ox \rho_{\R}  \}.
     \end{split}
 \end{equation}
Taking the asymptotic limit, the following error exponent of smooth conditional min-entropy~\cite{Li_2023} provides a converse bound on the error exponent for quantum state merging: 
\begin{equation}
\lim_{n\to\infty}
-\frac1n
\log
\inf\left\{
 \delta\in[0,1]
 \,\middle|\,
 H^\delta_{\rm min}
 (\A^n|\R^n)_{\rho^{\ox n}}
 \geq nr
\right\}
=
\frac12
\sup_{s\geq0}
s\left(
 \widetilde H_{1+s}(\A|\R)_\rho-r
\right).
\end{equation}
As a result, we have
\begin{equation}
    E^{\rm{merg},c}(\rho_{\A\B\R},r)  \leq \frac{1}{2}\sup_{s>0} s\left(\widetilde{H}_{1+s}(\A|\R)_\rho+r\right)=\frac{1}{2}\sup_{s>0} s\left(-H^{\uparrow}_{\frac{1}{1+s}}(\A|\B)_\rho+r\right),
\end{equation}
and
\begin{equation}
     E^{\rm{merg},d}(\rho_{\A\B\R},r)  \leq \frac{1}{2}\sup_{s>0} s\left(\widetilde{H}_{1+s}(\A|\R)_\rho-r\right)
     =\frac{1}{2}\sup_{s>0} s\left(-H^{\uparrow}_{\frac{1}{1+s}}(\A|\B)_\rho-r\right).
\end{equation}

The converse and achievability bounds coincide whenever
\[
r\leq
-\left.
\frac{\mathrm d}{\mathrm ds}
\left[
 s\widetilde H_{1+s}(\A|\R)_\rho
\right]
\right|_{s=1}
\]
in the entanglement-cost setting, and whenever
\[
r\geq
\left.
\frac{\mathrm d}{\mathrm ds}
\left[
 s\widetilde H_{1+s}(\A|\R)_\rho
\right]
\right|_{s=1}
\]
in the entanglement-distillation setting.


\section{Entanglement distillation}
\label{sec:entangle}

Let $\rho_{\A\B} \in \mc{S}(\A\B)$ be a bipartite quantum state shared by Alice and Bob.
In the resource theory of entanglement, entanglement distillation aims at transforming $\rho_{\A\B}$
into a maximally entangled state by means of a prescribed class of free operations.

We consider three classes of free operations: one-way LOCC, two-way LOCC, and non-entangling channels, i.e., quantum channels that preserve the set of separable states. We denote the corresponding sets by
$\mc{F}_{\rm{\rightarrow}}$, $\mc{F}_{\leftrightarrow}$ and $\mc{F}_{\rm{NE}}$, respectively, and define $\mc{G}:=\{\rightarrow,\leftrightarrow,\rm{NE}\}$.

For $t \in \mc{G}$, an $t$-assisted
entanglement distillation protocol for $\rho_{\A\B}$ consists of $(d,\mscr{L}_{\A:\B \rightarrow \M:\M'})$, where $\mscr{L}_{\A:\B \rightarrow \M:\M'} \in \mc{F}_t$, with $|\M|=|\M'|=d$. Let $\Phi_{\M\M'}$ denote the maximally entangled state of Schmidt rank $d$.
The error of this protocol is defined as
\[
p_{\rm{err}}^t(\mscr{L}\mid\rho_{\A\B}):=P(\mscr{L}(\rho_{\A\B}),\Phi_{\M\M'}).
\]
For a fixed $d=2^\lambda$, the optimal error among all protocols $(d,\mscr{L}_{\A:\B \rightarrow \M:\M'})$ is defined as 
\[
p_{\rm{err}}^t(\rho_{\A\B},\lambda):=\min_{(d, \mscr{L})} p_{\rm{err}}^t(\mscr{L}\mid\rho_{\A\B}).
\]
Since $\mc{F}_{\rightarrow} \subset \mc{F}_{\leftrightarrow} \subset \mc{F}_{\rm{NE}}$, it follows that for all $\rho_{\A\B}\in \mc{S}(\A\B)$ and $\lambda \geq 0$,
\begin{equation}
\label{equ:relation4}
    p_{\rm{err}}^{\rm{NE}}(\rho_{\A\B},\lambda) \leq p_{\rm{err}}^{\leftrightarrow}(\rho_{\A\B},\lambda) \leq p_{\rm{err}}^{\rightarrow}(\rho_{\A\B},\lambda).
\end{equation}

For a rate $r \geq 0$, the error exponent for $t$-assisted entanglement distillation is defined as 
\begin{equation}
    E^{t}(\rho_{\A\B},r):=\liminf_{n \rightarrow \infty} \frac{-1}{n} \log p_{\rm{err}}^{t}(\rho_{\A\B}^{\ox n},nr).
\end{equation}
In this section, we first establish achievability bounds on these error exponents. 

\begin{theorem}
\label{thm:entan}
    Let $\rho_{\A\B} \in \mc{S}(\A\B)$ and $r \geq 0$, for any $t \in \mc{G}\setminus \{\rm{NE}\}$, it holds that
    \begin{equation}
         E^{t}(\rho_{\A\B},r) \geq \frac{1}{2}\sup_{0<s<1} s\left(\lim_{m \rightarrow \infty}\tfrac{1}{m}\sup_{\mscr{L}_{\A^m:\B^m \rightarrow \C:\D}\in \mc{F}_t}I_{\frac{1}{1+s}}(\C \rangle \D)_{\mscr{L}(\rho_{\A\B}^{\ox m})}-r\right).
\end{equation}
Moreover,
\begin{equation}
         E^{\rm{NE}}(\rho_{\A\B},r) \geq \frac{1}{2}\sup_{0<s<1} s\left(\liminf_{m \rightarrow \infty}\tfrac{1}{m}\sup_{\mscr{L}_{\A^m:\B^m \rightarrow \C:\D}\in \mc{F}_{\rm{NE}}}I_{\frac{1}{1+s}}(\C \rangle \D)_{\mscr{L}(\rho_{\A\B}^{\ox m})}-r\right).
    \end{equation}
\end{theorem}
\begin{proof}
We first prove this claim for $t=\rightarrow$. 
Fix $m \in \mathbb{N}$ and a channel $\mscr{L}_{\A^m:\B^m \rightarrow \C:\D}\in \mc{F}_\rightarrow$, and define
$$\phi^{\mc{L}}_{\C\D}:=\mscr{L}(\rho_{\A\B}^{\ox m}).$$ 
Let $\phi^{\mc{L}}_{\C\D\R}$ be a purification of $\phi^{\mc{L}}_{\C\D}$.

By the proof of Theorem~\ref{thm:statemer}, there exists a one-way LOCC channel 
$$\Lambda_{\C:\D \rightarrow \bar{\M}_m:\tilde{\M}_m}=\sum_i \mscr{E}^i_{\C \rightarrow \bar{\M}_m} \ox \mscr{D}^i_{\D \rightarrow \tilde{\M}_m},$$
such that for any $s\in (0,1)$,
\begin{equation}
\label{equ:entdisapp}
\begin{split}
&P(\Lambda_{\C:\D \rightarrow \bar{\M}_m:\tilde{\M}_m}\circ\mscr{L}(\rho_{\A\B}^{\ox m}),\Phi_{\bar{\M}_m\tilde{\M}_m})\\
=&P(\Lambda_{\C:\D \rightarrow \bar{\M}_m:\tilde{\M}_m}(\phi^{\mc{L}}_{\C\D}),\Phi_{\bar{\M}_m\tilde{\M}_m}) \\
\leq  &\sqrt{2G_s 2^{mrs}2^{s\widetilde{D}_{1+s}(\phi^{\mc{L}}_{\C\R})\|I_{\C} \ox \phi^{\mc{L}}_{\R}))}},
    \end{split}
\end{equation}
where $\Phi_{\bar{\M}_m\tilde{\M}_m}$ is the maximally entangled state of Schmidt rank $|\bar{\M}_m|=2^{mr}$. 

Eq.~(\ref{equ:entdisapp}) implies that,
\begin{equation}
\label{equ:locgen}
    p^{\rightarrow}_{\rm{err}}(\rho_{\A\B}^{\ox m},mr) \leq \sqrt{2G_s 2^{mrs}2^{s \inf_{\mc{L} \in \mc{F}_t}\widetilde{D}_{1+s}(\phi^{\mc{L}}_{\C\R})\|I_{\C} \ox \phi^{\mc{L}}_{\R}))}}.
\end{equation}
Taking logarithms, dividing $m$, and letting $m \rightarrow \infty$ gives
\begin{equation}
         E^{\rightarrow}(\rho_{\A\B},r) \geq \frac{1}{2} s\left(\liminf_{m\rightarrow \infty}\tfrac{1}{m}\sup_{\mc{L} \in \mc{F}_t}I_{\frac{1}{1+s}}(\C \rangle \D)_{\mscr{L}(\rho_{\A\B}^{\ox m})}-r\right).
    \end{equation}
Since the classes $\mc F_{\rightarrow}$ are closed under tensor products and the Petz--R\'enyi coherent information is additive on product states, the sequence inside the regularization is superadditive. Hence its normalized limit exists and equals its supremum by Fekete's lemma. Then, Optimizing over $0<s<1$ yields the claimed bound for $t=\rightarrow$.

For $t=\leftrightarrow$ and $t=\rm{NE}$, the same argument applies since $\mc{F}_\rightarrow \subseteq \mc{F}_t$. 
In particular, the initial operation $\mscr{L}$ can be chosen from the larger class, 
while the one-way LOCC post-processing $\Lambda$ constructed above  also belongs to $\mc{F}_t$. 
This establishes the bound for all $t\in\mc{G}$. However, for $t=\rm{NE}$, the non-entangling class are not closed under tensor products in general, the Fekete's lemma can not be applied. Therefore, the $\liminf$ must be retained in this setting.

\end{proof}

The error exponent under non-entangling operations was characterized in~\cite{lin2026exponentialanalysisentanglementdistillation}. Specifically, for every $\rho_{\A\B}\in \mc{S}(\A\B)$ and $r > 0$,
\begin{equation}
\label{equ::renlim}
    E^{\rm{NE}}(\rho_{\A\B},r)=\frac{1}{2}\lim_{n \rightarrow \infty} \sup_{s>0} s\left( \tfrac{1}{n}\inf_{\sigma_{\A^n\B^n} \in \rm{SEP}(\A^n:\B^n)} D_{\frac{1}{1+s}}(\rho_{\A\B}^{\ox n}\|\sigma_{\A^n\B^n})-r\right),
\end{equation}
where $\rm{SEP}(\A:\B)$ denotes the set of separable states on $AB$.

Combining the above characterization with the monotonicity relation Eq.~(\ref{equ:relation4}) immediately yields the following converse bound.

\begin{corollary}
\label{cor:entan}
     Let $\rho_{\A\B} \in \mc{S}(\A\B)$ and $r > 0$, for any $t \in \mc{G}\setminus\{\rm{NE}\}$, it holds that
    \begin{equation}
         E^{t}(\rho_{\A\B},r) \leq \frac{1}{2}\lim_{n \rightarrow \infty} \sup_{s>0} s\left( \frac{1}{n}\inf_{\sigma_{\A^n\B^n} \in \rm{SEP}(\A^n:\B^n)} D_{\frac{1}{1+s}}(\rho_{\A\B}^{\ox n}\|\sigma_{\A^n\B^n})-r\right).
    \end{equation}
\end{corollary}

Next, we show that for maximally correlated states, the converse bound derived above matches the achievability bound when the rate $r$ is above a critical point.

\begin{theorem}
Let $\rho_{\A\B}$ be a maximally correlated state and $r> 0$. Then, for any $t \in \{\rightarrow,\leftrightarrow\}$,
\begin{equation}
\label{equ:entlowe}
    E^{t}(\rho_{\A\B},r) \leq \frac{1}{2} \sup_{s>0} s\left(I_{\frac{1}{1+s}}(\A\rangle \B)_\rho-r\right).
    \end{equation}
Moreover, if $r\geq \frac{\rm{d}}{\rm{ds}}sI_{\frac{1}{1+s}}(\A \rangle \B)_{\rho}|_{s=1}$, then
\begin{equation}
\label{equ:exactlocc}
     E^{t}(\rho_{\A\B},r)= \frac{1}{2}\sup_{0<s<1} s(I_{\frac{1}{1+s}}(\A\rangle \B)_\rho-r).
\end{equation}
\end{theorem}
\begin{proof}
    We first establish that for every $\alpha \in [0,2]$,
    \begin{equation}
    \label{equ:redata2}
\min_{\sigma_{\A\B} \in \rm{SEP}(\A:\B)}D_\alpha(\rho_{\A\B}\|\sigma_{\A\B}) = I_{\alpha}(\A \rangle \B)_{\rho}.
    \end{equation}
For any $\sigma_{\A\B} \in  \rm{SEP}(\A:\B)$,
\[
\sigma_{\A\B} \leq I_{A} \ox \sigma_{\B}.
\]
Hence,
\begin{equation}
\label{equ:relnew1}
    \min_{\sigma_{\A\B} \in \rm{SEP}(\A:\B)}D_\alpha(\rho_{\A\B}\|\sigma_{\A\B}) \geq \min_{\sigma_{\A\B} \in \rm{SEP}(\A:\B)}D_\alpha(\rho_{\A\B}\|I_A\ox \sigma_{\B}) \geq I_{\alpha}(\A \rangle \B)_{\rho}.
\end{equation}
Since $\rho_{\A\B}$ is a maximally correlated state, there exist orthonormal bases $\{\ket{a_x}\}_x$ and $\{\ket{b_x}\}_x$ of $\mc{H}_{\A}$ and $\mc{H}_{\B}$ such that $\supp(\rho_{\A\B}) \subset \text{span} \{\ket{a_x}\ox \ket{b_x} \}_x$.

Let $\Pi=\sum_x \proj{a_x}\ox \proj{b_x}$ be the corresponding projector. Consider the quantum channel
 $$\mscr{E}(\cdot)=\Pi(\cdot)\Pi+(I-\Pi)(\cdot)(I-\Pi).$$
 By data processing inequality for the Petz R\'enyi divergence, for any $\sigma_{\B} \in \mc{S}(\B)$,
\begin{equation}
\label{equ:reldata2}
\begin{split}
    &D_\alpha(\rho_{\A\B} \|I_{A} \ox \sigma_{\B}) \\
    \geq &D_\alpha(\mscr{E}(\rho_{\A\B}) \|\mscr{E}(I_{A} \ox \sigma_{\B})) \\
    =&D_\alpha(\rho_{\A\B} \|\Pi(I_{A} \ox \sigma_{\B})\Pi) \\
    =&D_\alpha(\rho_{\A\B} \|\sum_x \bra{b_x}\sigma_{\B}\ket{b_x} \proj{a_x}_{\A} \ox \proj{b_x}_{\B})\\
    \geq & \min_{\sigma_{\A\B} \in \rm{SEP}(\A:\B)}D_\alpha(\rho_{\A\B}\|\sigma_{\A\B}),
    \end{split}
\end{equation}
Eq.~(\ref{equ:relnew1}) and Eq.~(\ref{equ:reldata2}) give Eq.~(\ref{equ:redata2}). 

For maximally correlated states, the Petz R\'enyi relative entropy with respect to separable states is additive, i.e.,
\begin{equation}
   \min_{\sigma_{\A^n\B^n} \in \rm{SEP}(\A^n:\B^n)}D_\alpha(\rho_{\A\B}^{\ox n}\|\sigma_{\A^n\B^n}) =n\min_{\sigma_{\A\B} \in \rm{SEP}(\A:\B)}D_\alpha(\rho_{\A\B}\|\sigma_{\A\B}).
\end{equation}
Hence the regularized quantity in Corollary~\ref{cor:entan} reduces to the single-letter expression
$I_{\alpha}(\A \rangle \B)_{\rho}$, yielding Eq.~(\ref{equ:entlowe}).

Noticing that from the achievability bound (Theorem~\ref{thm:entan}), we have
\begin{equation}
\label{equ:loccachicov12}
      E^{t}(\rho_{\A\B},r)\geq \frac{1}{2}\sup_{0<s<1} s\left( I_{\frac{1}{1+s}}(\A \rangle \B)_{\rho}-r\right),
\end{equation}
 The exact characterization in the high-rate regime $r\geq \frac{\rm{d}}{\rm{ds}}sI_{\frac{1}{1+s}}(\A \rangle \B)_{\rho}|_{s=1}$ follows immediately by combining Eq.~(\ref{equ:loccachicov12}), Eq.~(\ref{equ:entlowe}) and a similar argument as in the proof of Theorem~\ref{thm:main}.
\end{proof}


\section{Quantum communication}
\label{sec:quantumcom}

In this section, we investigate several fundamental quantum communication tasks over a quantum channel
$\mscr{N}_{\A \rightarrow \B}$. Throughout, performance is quantified using the purified distance.

\begin{itemize}
    \item \emph{Subspace transmission:} An subspace transmission protocol $(d, \mscr{E}_{\M \rightarrow \A}, \mscr{D}_{\B \rightarrow \B'})$ consists of an encoding channel $\mscr{E}_{\M \rightarrow \A}$ and a decoding channel $\mscr{D}_{\B \rightarrow \bar{\M}}$, with $|\M|=|\bar{\M}|=d$.
    
The overall protocol induces the effective channel
\[
\mscr{M}_{\M \rightarrow \bar{\M}}:=\mscr{D}_{\B \rightarrow \bar{\M}} \circ \mscr{N}_{\A \rightarrow \B} \circ \mscr{E}_{\M\rightarrow \A}.
\]
The goal is to approximate the identity channel $id_{\M \rightarrow \bar{\M}}$. The error is defined as
\[
p_{\rm{err}}^{\rm{QC}}(\mscr{E},\mscr{D};\mscr{N}):=\max_{\phi_{\R\M}}P( \mscr{M}_{\M \rightarrow \bar{\M}}(\phi_{\R\M}) , \phi_{\R\bar{\M}} ),
\]
where the maximization is over all pure states $\phi_{\R\M}$. 

\item \emph{One-way LOCC-assisted subspace transmission:} A one-way LOCC-assisted subspace transmission protocol $(d,\mscr{E}_{\M \rightarrow \A\bar{\A}},\mscr{D}_{\bar{\A}:\B \rightarrow \varnothing : \bar{\M}})$ consists of an encoding channel $\mscr{E}_{\M \rightarrow \A\bar{\A}}$ and a one-way LOCC decoding channel $\mscr{D}_{\bar{\A}:\B \rightarrow \varnothing : \bar{\M}}$ from Alice to Bob, with $|\M|=|\bar{\M}|=d$.

The protocol induces an effective channel
\[
\mscr{M}_{\M \rightarrow \bar{\M}}:=\mscr{D}_{\bar{\A}:\B \rightarrow \varnothing : \bar{\M}} \circ \mscr{N}_{\A \rightarrow \B} \circ\mscr{E}_{\M\rightarrow \A\bar{\A} }.
\]
and the error $p_{\rm{err}}^{(\rm{QC},\rightarrow)}(\mscr{E},\mscr{D};\mscr{N})$ is defined analogously to the subspace transmission setting, namely as the maximum purified distance between the output state and the target state over all pure inputs.

\item \emph{Entanglement transmission:} An entanglement transmission protocol $(d,\mscr{E}_{\M \rightarrow \A},\mscr{D}_{\B \rightarrow \bar{\M}})$ is defined by three elements, in which $\mscr{E}_{\M \rightarrow \A}$ is an encoding channel and $\mscr{D}_{\B \rightarrow \bar{\M}}$ is a decoding channel, with $|\M|=|\bar{\M}|=d$.

The goal is to transmit one half of a maximally entangled state $\Phi_{\R\M}$ from Alice to Bob. The output state
\[
\omega_{\R\bar{\M}}:=(\mscr{D}_{\B \rightarrow \bar{\M}} \circ \mscr{N}_{\A \rightarrow \B} \circ \mscr{E}_{\M\rightarrow \A }) (\Phi_{\R\M})
\]
is required to approximate $\Phi_{\R\bar{\M}}$. The entanglement transmission
error of the protocol is
\[
p_{\rm{err}}^{\rm{ET}}(\mscr{E},\mscr{D};\mscr{N}):=P(\omega_{\R\bar{\M}},\Phi_{\R\bar{\M}}).
\]
\item \emph{One-way LOCC-assisted entanglement transmission:}  An one-way LOCC-assisted entanglement transmission protocol $(d,\mscr{E}_{\M \rightarrow \A\bar{\A}},\mscr{D}_{\bar{\A}:\B \rightarrow \varnothing: \bar{\M}})$ consists of an encoding channel $\mscr{E}_{\M \rightarrow \A}$  and a one-way LOCC decoding channel $\mscr{D}_{\bar{\A}:\B \rightarrow \varnothing: \bar{\M}}$ from Alice to Bob, with $|\M|=|\bar{\M}|=d$. The goal of the protocol is to transmit the $\M$ system of a maximally entangled state $\Phi_{\R\M}$ such that the final state
\[
\omega_{\R\bar{\M}}:=(\mscr{D}_{\bar{\A}:B \rightarrow \varnothing:\bar{\M}} \circ \mscr{N}_{\A \rightarrow \B} \circ \mscr{E}_{\M\rightarrow \A\bar{\A}}) (\Phi_{\R\M})
\]
is close to the initial maximally entangled state. The 
error of the protocol is defined as 
\[
p_{\rm{err}}^{(\rm{ET}, \rightarrow)}(\mscr{E},\mscr{D};\mscr{N}):=P(\omega_{\R\bar{\M}},\Phi_{\R\bar{\M}}).
\]
\item \emph{Entanglement generation:}  An entanglement generation protocol is defined by three elements $(d, \Psi_{\M'\A}, \mscr{D}_{\B \rightarrow \bar{\M}}) $, where $\Psi_{\M'\A}$ is a pure state and $\mscr{D}_{\B \rightarrow \bar{\M}}$ is a decoding channel, with $|\bar{\M}|=|\M'|=d$. 

The output state
\[
\sigma_{\M'\bar{\M}}:=(\mscr{D}_{\B \rightarrow \bar{\M}}\circ \mscr{N}_{\A \rightarrow \B})(\Psi_{\M'\A})
\]
is required to approximate a maximally entangled state $\Phi_{\M'\bar{\M}}$.  The
entanglement generation error of the protocol is given by
\[
p_{\rm{err}}^{\rm{EG}}(\Psi_{\M'\A}, \mscr{D};\mscr{N}):=P(\sigma_{\M'\bar{\M}},\Phi_{\M'\bar{\M}}).
\]

\item \emph{One-way LOCC-assisted entanglement generation:}  An one-way LOCC-assisted entanglement generation protocol consists of a triplet $(d,\Psi_{\A'\A},\{\mscr{E}^x_{\A'\A \rightarrow \M'\A}\}_x,\{\mscr{D}^x_{\B \rightarrow \bar{\M}}\}_x)$, where $\Psi_{\A'\A}$ is a pure state with $|M'|=d$, $\{\mscr{E}^x_{\A'\A \rightarrow \M'\A}\}_x$ is a set of completely positive maps indexed by a finite alphabet $\mc{X}$ such that $\sum_{x} \mscr{E}^x_{\A'\A \rightarrow \M'\A}$ is trace preserving, $\{\mscr{D}^x_{\B \rightarrow \bar{\M}}\}_x$ is a
set of quantum channels indexed by $\mc{X}$, with $|\bar{\M}|=d$. The goal of the protocol
is to transmit the system $A$ such that the final state
\[
\sigma_{\M'\bar{\M}}^\rightarrow:=\sum_x (\mscr{D}^x_{\B \rightarrow \bar{\M}} \circ \mscr{N}_{\A \rightarrow \B} \circ \mscr{E}^x_{\A'\A \rightarrow \M'\A})(\Psi_{\A'\A})
\]
is close to a maximally entangled state of Schmidt rank $d$. The error
of the protocol is defined as 
\[
p_{\rm{err}}^{\rm{EG, \rightarrow}}(\Psi_{\A'\A},\{\mscr{E}_x \}_x, \{\mscr{D}_X\}_x;\mscr{N}):=P(\sigma_{\M'\bar{\M}}^\rightarrow,\Phi_{\M'\bar{\M}}).
\]
\end{itemize}
Let 
\[
\mc{T}:=\{\rm{ST},\rm{ET},\rm{EG},(\rm{ST}, \rightarrow), (\rm{ET},\rightarrow), (\rm{EG},\rightarrow)\}
\]
denote the set of this collection of tasks.

For $\epsilon \in [0,1]$, a protocol is called $(d,\epsilon)$-achievable for task $t \in \mc{T}$ if its error does not exceed $\epsilon$. 
For $t \in \mc{T}$, the minimal achievable error among all protocols with fixed $d=2^\lambda$ is defined as
\[
p_{\rm{err}}^{t}(\mscr{N}_{\A \rightarrow \B},\lambda):=\inf\{\epsilon~|~\exists~(2^\lambda,\epsilon)\text{-achievable protocol for task}~t\}.
\]
From the above definitions, we can directly obtain that for any quantum channel $\mscr{N}_{\A \rightarrow \B}$ and $\lambda \geq 0$,
\begin{align}
\label{equ:relation1}
    &p_{\rm{err}}^{\rm{ET},\rightarrow}(\mscr{N}_{\A \rightarrow \B},\lambda) \leq p_{\rm{err}}^{\rm{ST},\rightarrow}(\mscr{N}_{\A \rightarrow \B},\lambda) \leq p_{\rm{err}}^{\rm{ST}} (\mscr{N}_{\A \rightarrow \B},\lambda),\\
\label{equ:relation2}
    &p_{\rm{err}}^{\rm{ET},\rightarrow}(\mscr{N}_{\A \rightarrow \B},\lambda) \leq p_{\rm{err}}^{\rm{ET}}(\mscr{N}_{\A \rightarrow \B},\lambda) \leq p_{\rm{err}}^{\rm{ST}}(\mscr{N}_{\A \rightarrow \B},\lambda).
\end{align}
Moreover, by Lemma~\ref{lem:wilde1}, Lemma~\ref{lem:wilde2}, 
\begin{align}
\label{equ:relation3}
    p_{\rm{err}}^{\rm{EG},\rightarrow} (\mscr{N}_{\A \rightarrow \B},\lambda)= p_{\rm{err}}^{\rm{EG}}(\mscr{N}_{\A \rightarrow \B},\lambda) \leq p_{\rm{err}}^{\rm{ET}}(\mscr{N}_{\A \rightarrow \B},\lambda) \leq 2 p_{\rm{err}}^{\rm{EG}}(\mscr{N}_{\A \rightarrow \B},\lambda).
\end{align}
and by Lemma~\ref{lem:wilde3},
\begin{align}
\label{equ:relationnew}
    p_{\rm{err}}^{\rm{ST}}(\mscr{N}_{\A \rightarrow \B},\lambda) \leq \sqrt{2}p_{\rm{err}}^{\rm{ET}}(\mscr{N}_{\A \rightarrow \B},\lambda+1).
\end{align}
For a fixed rate $r$, the error exponent of task $t \in \mc{T}$ is defined as
\[
E^{t}(\mscr{N}_{\A \rightarrow \B},r):=\liminf_{n \rightarrow \infty} \frac{-1}{n} \log p_{\rm{err}}^{t}(\mscr{N}_{\A \rightarrow \B}^{\ox n}, nr).
\]
In the following, we study these error exponents and first establish an achievability bound.

\begin{theorem}
\label{thm:loccachi}
    Let $\mscr{N}_{\A \rightarrow \B}$ be a quantum channel and  $r \geq 0$. For any task $t \in \mc{T}$, we have
    \begin{equation}
    \label{equ:locc}
E^{t}(\mscr{N}_{\A \rightarrow \B},r) \geq \frac{1}{2}\sup_{0<s<1} s\left(\lim_{m \rightarrow \infty}\tfrac{1}{m}I_{\frac{1}{1+s}}^{\rm{c}}(\mscr{N}_{\A \rightarrow \B}^{\ox m})-r\right),
\end{equation}
where $I_\alpha^{\rm{c}}(\mscr{N}_{\A \rightarrow \B}):=\max_{\phi_{\bar{\A}\A}}I_{\alpha}(\bar{\A} \rangle \B)_{\mscr{N}(\phi_{\bar{\A}\A})}$ is the Petz R\'enyi coherent information of $\mscr{N}_{\A \rightarrow \B}$.
\end{theorem}
\begin{proof}
Fix a integer $m$ and a pure state $\phi_{\bar{\A}^m\A^m}$. Let $$\phi_{\bar{\A}^m\B^m\E^m}=V^{\ox m}_{\A \rightarrow \B\E}\phi_{\bar{\A}^m\A^m}V^{* \ox m}_{\A \rightarrow \B\E}$$ be the purification of the channel output, where $V_{\A \rightarrow \B\E}$ is the Stinespring isometry of $\mscr{N}_{\A \rightarrow \B}$. 

By the proof of Theorem~\ref{thm:statemer}, for any $k \in \mathbb{N}$, there exists an one-way LOCC channel 
$$\Lambda_{\bar{\A}^{mk}:\B^{mk} \rightarrow \bar{\M}_k:\tilde{\M}_k}=\sum_i \mscr{E}^i_{\bar{\A}^{mk} \rightarrow \bar{\M}_k} \ox \mscr{D}^i_{\B^{mk} \rightarrow \tilde{\M}_k},$$
such that 
\begin{equation}
\label{equ:entdisapp1}
    P(\Lambda_{\bar{\A}^{mk}:\B^{mk} \rightarrow \bar{\M}_k:\tilde{\M}_k}(\phi_{\bar{\A}^m\B^m}^{\ox k}),\Phi_{\bar{\M}_k\tilde{\M}_k}) \leq  \sqrt{2G_s 2^{mkrs}2^{ks\widetilde{D}_{1+s}(\phi_{\bar{\A}^m\E^m}\|I_{\bar{\A}}^{\ox m} \ox \phi_{\E^m})}},
\end{equation}
where $\Phi_{\bar{\M}_k\tilde{\M}_k}$ is the maximally entangled state of Schmidt rank $|\bar{\M}_k|=2^{mkr}$. 

Eq.~(\ref{equ:entdisapp1}) implies that for $n=mk$,
\begin{equation}
\label{equ:locgen1}
    p_{\rm{err}}^{\rm{EG},\rightarrow}(\mscr{N}^{\ox n}_{\A \rightarrow \B},nr) \leq \sqrt{2G_s2^{mkrs}2^{ks\widetilde{D}_{1+s}(\phi_{\bar{\A}^m\E^m}\|I_{\bar{\A}}^{\ox m} \ox \phi_{\E^m})}}.
\end{equation}
Taking logarithms and letting $n \rightarrow \infty$, optimizing over $m$, pure state $\phi_{\bar{\A}^m\A^m}$ and $0<s<1$ yields the bound for
 for $t=(\rm{EG},\rightarrow)$. By Eq.~(\ref{equ:relation3}) and Eq.~(\ref{equ:relationnew}), this extends to $t\in \{\rm{EG}, \rm{ET},\rm{ST}\}$.

Next, we prove Eq.~(\ref{equ:locc}) for $t \in \{(\rm{ST},\rightarrow), (\rm{ET},\rightarrow)  \}$. Using the standard teleportation protocol, the maximally entangled state $\Phi_{\bar{\M}_k\tilde{\M}_k}$ in Eq.~(\ref{equ:entdisapp1}) can be converted via LOCC into a noiseless quantum channel. Specifically, there exists an one-way LOCC channel from Alice to Bob
$$\Pi_{\bar{\M}_k\M_k:\tilde{\M}_k \rightarrow \varnothing: \M'_k}$$
such that 
\[
\Pi_{\bar{\M}_k\M_k:\tilde{\M}_k \rightarrow \varnothing :\M'_k}((\cdot)\ox \Phi_{\bar{\M}_k\tilde{\M}_k})=\id_{\M_k \rightarrow \M'_k}(\cdot).
\]
Now for $n=mk$, we construct an one-way LOCC-assisted quantum communication protocol for $\mscr{N}^{\ox n}_{\A \rightarrow \B}$ as follows. The encoder prepares the fixed state $\phi^{\ox k}_{\bar{\A}^m\A^m}$, while the decoder consists of the concatenation of $\Lambda_{\bar{\A}^{mk}:\B^{mk} \rightarrow \bar{\M}_k:\tilde{\M}_k}$ and the teleportation channel $\Pi$.

The resulting error satisfies 
\begin{equation}
\label{equ:evaqu}
    \begin{split}
       &p_{\rm{err}}^{\rm{ST}, \rightarrow}(\mscr{N}^{\ox n}_{\A \rightarrow \B},nr)\\
        \leq& \max_{\phi_{\M_k\R_k}}P(\Pi_{\bar{\M}_k\M_k:\tilde{\M}_k \rightarrow \varnothing: \M'_k}(\Lambda_{\bar{\A}^{mk}:\B^{mk} \rightarrow \bar{\M}_k:\tilde{\M}_k}(\phi_{\bar{\A}^m\B^m}^{\ox k})\ox \phi_{\M_k\R_k}),id_{\M_k \rightarrow \M'_k}(\phi_{\M_k\R_k})) \\
        =&  \max_{\phi_{\M_k\R_k}}P(\Pi_{\bar{\M}_k\M_k:\tilde{\M}_k \rightarrow \varnothing :\M'_k}(\Lambda_{\bar{\A}^{mk}:\B^{mk} \rightarrow \bar{\M}_k:\tilde{\M}_k}(\phi_{\bar{\A}^m\B^m}^{\ox k})\ox \phi_{\M_k\R_k}),\Pi_{\bar{\M}_k\M_k:\tilde{\M}_k \rightarrow \varnothing :\M'_k}(\phi_{\M_k\R_k}\ox \phi_{\bar{\M}_k\tilde{\M}_k})) \\
        \leq & \max_{\phi_{\M_k\R_k}} P(\Lambda_{\bar{\A}^{mk}:\B^{mk} \rightarrow \bar{\M}_k:\tilde{\M}_k}(\phi_{\bar{\A}^m\B^m}^{\ox k})\ox \phi_{\M_k\R_k},\Phi_{\bar{\M}_k\tilde{\M}_k}\ox\phi_{\M_k\R_k} ) \\
        =&P\left(\Lambda_{\bar{\A}^{mk}:\B^{mk} \rightarrow \bar{\M}_k:\tilde{\M}_k}(\phi_{\bar{\A}^m\B^m}^{\ox k}), \Phi_{\bar{\M}_k\tilde{\M}_k}\right)\\
        \leq &  \sqrt{2G_s 2^{mkrs}2^{ks\widetilde{D}_{1+s}(\phi_{\bar{\A}^m\E^m}\|I_{\bar{\A}}^{\ox m} \ox \phi_{\E^m})}}.
    \end{split}
\end{equation}
so the same bound as in Eq.~(\ref{equ:locgen1}) applies. Optimizing as above yields the bound for $t=(\rm{ST},\rightarrow)$, and Eq.~(\ref{equ:relation1}) gives the result for $t=(\rm{ET},\rightarrow)$.

The construction above was stated for blocklengths $n=mk$. For an
arbitrary sufficiently large $n=mk+\ell$, $0\leq\ell<m$, we run the
$mk$-copy protocol at rate $r+\delta$ and restrict the resulting
maximally entangled state or code space to dimension $2^{nr}$.
For all sufficiently large $k$, $mk(r+\delta)\geq nr$, and this
restriction cannot increase the purified-distance error. Letting
$\delta\downarrow0$ gives the claimed liminf bound for all
blocklengths. Moreover,
$I_\alpha^{\rm c}(\mscr N^{\ox m})$ is superadditive in $m$, so the
regularized limit exists by Fekete's lemma.
\end{proof}

Since the Petz R\'enyi coherent information is not additive in general, the regularization appearing in Theorem~\ref{thm:loccachi} is in general unavoidable. However, for covariant quantum channels with respect to a group whose representation on
$\A$ forms a unitary one-design, we obtain a single-letter converse bound on the error exponent.

We begin with a one-shot upper bound on the quantity $R^{\rm{ET},\rightarrow}(\mscr{N}_{\A \rightarrow \B},\epsilon)$, which denotes the maximal number of qubits that can be transmitted with error at most $\epsilon \in (0,1)$ under one-way LOCC assistance. For any covariant quantum channel $\mscr{N}_{A\to \B}$ with respect to a group whose representation on
$\A$ forms a unitary one-design and any $n\in\mathbb{N}$, it was shown in~\cite{Tomamichel_2016} that 
\begin{equation}
\label{equ:locchy}
    R^{\rm{ET},\rightarrow}(\mscr{N}_{\A \rightarrow \B}^{\ox n},\epsilon)\leq \min_{\sigma_{\A\B} \in \rm{PPT}'(\A:\B)} D_{\rm{H}}^{\epsilon^2}(\mscr{N}_{\A' \rightarrow \B}(\Phi_{\A\A'})^{\ox n}\|\sigma_{\A\B}^{\ox n}),
\end{equation}
where 
\[
{\rm PPT}'(\A:\B)
:=
\left\{
 \tau_{\A\B}\in\mc P(\A\B)
 \,\middle|\,
 \left\|\tau_{\A\B}^{\top_{\B}}\right\|_1\leq1
\right\}.
\]
denotes the Rains set and $$D_{\rm{H}}^\epsilon(\rho\|\sigma):=-\log \min \{\tr \sigma T ~|~\tr \rho T\geq 1-\epsilon~\land~0\leq T \leq I \}$$ is the hypothesis testing relative entropy. Since
\[
\rm{PPT}(\A:\B):=\left\{
 \sigma_{\A\B}\in\mc S(\A\B)
 \,\middle|\,
 \sigma_{\A\B}^{\top_{\B}} \geq 0,
\right\}   \subseteq \rm{PPT}'(\A:\B) 
\]
we have
\begin{equation}
    \label{equ:locchy1}
    R^{\rm{ET},\rightarrow}(\mscr{N}_{\A \rightarrow \B}^{\ox n},\epsilon)\leq \min_{\sigma_{\A\B} \in \rm{PPT}(\A:\B)} D_{\rm{H}}^{\epsilon^2}(\mscr{N}_{\A' \rightarrow \B}(\Phi_{\A\A'})^{\ox n}\|\sigma_{\A\B}^{\ox n}),
\end{equation}
Eq.~(\ref{equ:locchy1}) implies that
\begin{equation}
\label{equ:horbas}
    \begin{split}
&p_{\rm{err}}^{\rm{ET},\rightarrow}(\mscr{N}^{\ox n}_{\A \rightarrow \B}, nr) \\
=&\inf \left\{\epsilon~\middle|~ R^{\rm{ET},\rightarrow}(\mscr{N}_{\A \rightarrow \B}^{\ox n},\epsilon) \geq nr \right\} \\
\geq&\inf \left\{\epsilon~\Big|~ \min_{\sigma_{\A\B} \in \rm{PPT}(\A:\B)}D_{\rm{H}}^{\epsilon^2}(\mscr{N}_{\A' \rightarrow \B}(\Phi_{\A\A'})^{\ox n}\|\sigma_{\A\B}^{\ox n}) \geq nr \right\} \\
=&\max_{\substack{
 \sigma_{\A\B}\in{\rm PPT}(\A:\B)
}}
2^{-\frac12
D_{\rm H}^{\,2^{-nr}}
\left(
 \left(
  \sigma_{\A\B}
 \right)^{\ox n}
 \middle\|
 \mscr N_{\A'\to\B}(\Phi_{\A\A'})^{\ox n}
\right)}.
    \end{split}
\end{equation}
For every fixed
$\sigma_{\A\B}\in{\rm PPT}(\A:\B)$, the quantum Hoeffding bound gives
\[
\begin{aligned}
&\lim_{n\to\infty}
\frac1n
D_{\rm H}^{\,2^{-nr}}
\left(
 \left(
  \sigma_{\A\B}
 \right)^{\ox n}
 \middle\|
 \mscr N_{\A'\to\B}(\Phi_{\A\A'})^{\ox n}
\right)\\
&\qquad=
\sup_{s>0}
s\left(
 D_{\frac1{1+s}}
 \left(
  \mscr N_{\A'\to\B}(\Phi_{\A\A'})
  \middle\|
  \sigma_{\A\B}
 \right)
 -r
\right).
\end{aligned}
\]
Furthermore, for every finite $S>0$, the function
\[
(\sigma_{\A\B},s)
\longmapsto
s\left(
 D_{\frac1{1+s}}
 \left(
  \mscr N_{\A'\to\B}(\Phi_{\A\A'})
  \middle\|
  \sigma_{\A\B}
 \right)
-r
\right)
\]
is convex in $\sigma_{\A\B}$ and concave in $s\in[0,S]$.
A standard regularization by a full-rank element of
${\rm PPT}(\A:\B)$ and Sion's minimax theorem therefore give
\begin{equation}
\label{equ:sion}
\inf_{\sigma_{\A\B}\in{\rm PPT}(\A:\B)}
\sup_{s>0}s\left( D_{\frac{1}{1+s}}(\mscr{N}(\Phi_{\A\A'})\|\sigma_{\A\B})-r\right)
=
\sup_{s>0}
\inf_{\sigma_{\A\B}\in{\rm PPT}(\A:\B)}
s\left( D_{\frac{1}{1+s}}(\mscr{N}(\Phi_{\A\A'})\|\sigma_{\A\B})-r\right).
\end{equation}

Applying the quantum Hoeffding exponent of quantum hypothesis testing to Eq.~(\ref{equ:horbas}) and combining Eq.~(\ref{equ:relation1}), Eq.~(\ref{equ:relation2}), Eq.~(\ref{equ:relation3}) and Eq.~(\ref{equ:sion}), we obtain the following conclusion.

\begin{corollary}
\label{cor:conlocc}
    For any covariant quantum channel $\mscr{N}_{\A \rightarrow \B}$  with respect to a group whose representation on
$\A$ forms a unitary one-design, $r >0$ and $t \in \mc{T}$, 
    \begin{equation}
    \begin{split}
        E^{t}(\mscr{N}_{\A \rightarrow \B},r) &\leq \sup_{0<\alpha<1} \frac{1-\alpha}{2\alpha} \left( \min_{\sigma_{\A\B} \in {\rm{PPT}}(\A:\B)}D_\alpha(\mscr{N}_{\A' \rightarrow \B}(\Phi_{\A\A'}))\|\sigma_{\A\B})-r\right) \\
        &=\frac{1}{2}\sup_{0<s<\infty} s\left( \min_{\sigma_{\A\B} \in {\rm{PPT}}(\A:\B)}D_{\frac{1}{1+s}}(\mscr{N}_{\A' \rightarrow \B}(\Phi_{\A\A'})\|\sigma_{\A\B})-r\right).
        \end{split}
    \end{equation}
\end{corollary}

Compared to the achievability bound, the converse bound involves only a single-letter quantity, making its expression more concise. The trade-off is that the converse bound is valid only for covariant quantum channels and is generally not tight. Nevertheless, for covariant generalized dephasing channels, we can show that the converse bound coincides with the achievability bound for coding rates above a critical value.

A covariant generalized dephasing channel is a generalized dephasing
channel of the form
\begin{equation}
\label{equ:cgdcform}
\mscr N_{\A\to\B}(\rho)
=
\sum_{x,y=0}^{|\A|-1}
\bra{x}\rho\ket{y}
\langle\psi_y|\psi_x\rangle
\ket{x}\!\bra{y}_{\B}
\end{equation}
that is covariant with respect to a group whose representation on
$\A$ forms a unitary one-design.

For such channels, we obtain the following result.

\begin{theorem}
Let $\mscr{N}_{\A \rightarrow \B}$ be a covariant generalized dephasing channel and $r>0$. Then, for any $t \in \mc{T}$,
    \begin{equation}
    \label{equ:loccconver}
    \begin{split}
        E^{t}(\mscr{N}_{\A \rightarrow \B},r)
        &\leq \frac{1}{2}\sup_{0<s<\infty} s\left( I_{\frac{1}{1+s}}(\A \rangle \B)_{\mscr{N}(\Phi_{\A\A'})}-r\right).
        \end{split}
    \end{equation}
Moreover, when $r\geq \frac{\rm{d}}{\rm{ds}}sI_{\frac{1}{1+s}}(\A \rangle \B)_{\mscr{N}(\Phi_{\A\A'})}|_{s=1}$, for any $t \in \mc{T} $, we have
\begin{equation}
\label{equ:exactlocc1}
     E^{t}(\mscr{N}_{\A \rightarrow \B},r)= \frac{1}{2}\sup_{0<s<1} s( I_{\frac{1}{1+s}}(\A \rangle \B)_{\mscr{N}(\Phi_{\A\A'})}-r).
\end{equation}
\end{theorem}
\begin{proof}
    We first establish that, for any $\alpha \in [0,2]$,
    \begin{equation}
    \label{equ:redata}
\min_{\sigma_{\A\B} \in {\rm{PPT}}(\A:\B)}D_\alpha(\mscr{N}(\Phi_{\A\A'})\|\sigma_{\A\B}) \leq I_{\alpha}(\A \rangle \B)_{\mscr{N}(\Phi_{\A\A'})}.
    \end{equation}
Since $\mscr{N}_{\A \rightarrow \B}$ has the form in Eq.~(\ref{equ:cgdcform}), the support of $\mscr{N}(\Phi_{\A\A'})$ lies in the projector 
$$\Pi=\sum_{x=0}^{|\A|-1} \proj{x}_{\A} \ox \proj{x}_{\B}.$$
Applying the data-processing inequality for the quantum channel $\mscr{E}(\cdot)=\Pi(\cdot)\Pi+(I-\Pi)(\cdot)(I-\Pi)$, we obtain, for any $\sigma_{\B} \in \mc{S}(\B)$,
\begin{equation}
\label{equ:reldata}
\begin{split}
    &D_\alpha(\mscr{N}(\Phi_{\A\A'}) \|I_{\A} \ox \sigma_{\B}) \\
    \geq &D_\alpha(\mscr{E}(\mscr{N}(\Phi_{\A\A'})) \|\mscr{E}(I_{\A} \ox \sigma_{\B})) \\
    =&D_\alpha(\mscr{N}(\Phi_{\A\A'}) \|\Pi(I_{\A} \ox \sigma_{\B})\Pi) \\
    =&D_\alpha(\mscr{N}(\Phi_{\A\A'}) \big\|\sum\nolimits_{x=0}^{|\A|-1} \bra{x}\sigma_{\B}\ket{x} \proj{x}_A \ox \proj{x}_B)\\
    \geq &\min_{\sigma_{\A\B} \in {\rm{PPT}}(\A:\B)}D_\alpha(\mscr{N}(\Phi_{\A\A'})\|\sigma_{\A\B}),
    \end{split}
\end{equation}
which proves Eq.~(\ref{equ:redata}) and Eq.~(\ref{equ:loccconver}) follows from Eq.~(\ref{equ:redata}) and Corollary~\ref{cor:conlocc}.

Noticing that the achievability bound (Theorem~\ref{thm:loccachi}) gives
\begin{equation}
\label{equ:loccachicov}
      E^{t}(\mscr{N}_{\A \rightarrow \B},r)\geq \frac{1}{2}\sup_{0<s<1} s\left( I_{\frac{1}{1+s}}(\A \rangle \B)_{\mscr{N}(\Phi_{\A\A'})}-r\right),
\end{equation}
 where $t \in \mc{T} $. The exact characterization in the high-rate regime $r\geq \frac{\rm{d}}{\rm{ds}}sI_{\frac{1}{1+s}}(\A \rangle \B)_{\mscr{N}(\Phi_{\A\A'})}|_{s=1}$ follows immediately by combining Eq.~(\ref{equ:loccachicov}), Eq.~(\ref{equ:loccconver}) and a similar argument as in the proof of Theorem~\ref{thm:main}.
\end{proof}


\section{Conclusion and Discussion}
\label{sec:conclu}

We derived fundamental one-shot upper bounds on the error of quantum information decoupling, together with a complementary one-shot lower bound in the standard decoupling setting. These bounds induce corresponding achievability and ensemble-tight converse bounds on the decoupling error exponent. By showing that the two bounds coincide whenever the decoupling rate $r$ satisfies $r \leq R_{\text{critical}}$, we obtained an exact characterization of the error exponent of standard quantum information decoupling in the low-rate regime.

As applications, we characterized the exact error exponents of quantum state merging in the low-entanglement-cost and high-entanglement-distillation-rate regimes and established achievability bounds on the error exponents for entanglement distillation and a broad class of quantum communication tasks, including subspace transmission, entanglement transmission, entanglement generation, as well as their one-way LOCC-assisted counterparts. We expect that the techniques and results developed here will find further applications in other quantum information processing tasks.

Our analysis also reveals the existence of a critical rate for the error exponents of standard decoupling and quantum state merging, beyond which the exact error exponents remain unknown. The appearance of such a critical rate, separating regimes of complete and partial characterization, is a common phenomenon in the study of error exponents~\cite{gallager1968information,Dalai_2013,Li_2023}. In classical channel coding, combinatorial effects are known to play a decisive role in this regime~\cite{gallager1968information,Dalai_2013}, and it remains an open question whether analogous mechanisms govern quantum information decoupling and quantum state merging.

Throughout this work, we quantified the decoupling error using quantum relative entropy. A natural direction for future research is to determine the corresponding error exponents under alternative distance measures, such as the trace distance. While prior work has established achievability bounds in trace distance, and related techniques\,---\,such as complex interpolation\,---\,have successfully been applied to other quantum covering problems, including quantum soft covering~\cite{cheng2023error}, convex splitting~\cite{cheng2023tightoneshotanalysisconvex}, and privacy amplification~\cite{dupuis2023privacy}, the development of matching converse bounds in fully quantum settings remains open.

Finally, our achievability bounds on the error exponents for entanglement distillation, subspace transmission, entanglement transmission, entanglement generation, and their one-way LOCC-assisted counterparts are expressed in terms of regularized Petz--R\'enyi coherent informations. Establishing concise single-letter characterizations or matching converse bounds for these tasks constitutes an important open problem in quantum information theory.


\acknowledgements

MB and YY acknowledge support by the European Research Council (ERC Grant Agreement No.~948139) and from the Excellence Cluster Matter and Light for Quantum Computing (ML4Q-2). HC is supported under grants No.~NSTC 114-2628-E-002 -006, NSTC 114-2119-M-001-002, and NSTC 115-2124-M-002-014, NTU-114V2016-1, NTU-CC115L893705, and NTU-115L900702.
We thank Bartosz Regula for providing insightful comments and suggestions to our early version
We thank Hayata Yamasaki for pointing out a problem around an non-full support operator in an earlier version of the proof of Theorem~\ref{thm:maindec} and for providing useful suggestions.


\appendix
\section{Auxiliary Lemmas} \label{sec:lemmas}

The following lemmas are used in our analysis, with notation as in the main text.

\begin{lemma}[\hspace{-0.1pt}{\cite{sharp25, gour2026optimaltraceinequalitiessingleshot}}] \label{lemm:sharp_one-shot}
For $\rho,\sigma\in\mc{P}(\mc{H})$ such that $\supp{(\rho)} \subseteq \supp{(\sigma)}$ and for any $s\in(0,1]$,
\begin{align}
	\tr\left( \rho \left( \log(\rho+\sigma) - \log \sigma \right) \right)
    &\leq G_s\tr\left( \left( \sigma^{-\frac{s}{2(1+s)}} \rho \sigma^{-\frac{s}{2(1+s)}} \right)^{1+s} \right).
    \label{eq:sharp_one-shot}
\end{align}
where $G_s:=\sup_{r>0} \frac{\log(1+r)}{r^s}  \leq \frac{1}{\ln 2} s^{s-1}(1-s)^{1-s}$ for $s\in(0,1]$.
\end{lemma}

\begin{lemma} [\hspace{-0.1pt}{\cite[Lemma 1]{Li_2023}}]
\label{lem:con}
    Let $A_x \in \mc{P}(\mc{H})$ for $x \in \mc{X}$, and let $\lambda$ be a positive number. Then we have
    \begin{equation}
        \tr\left(\sum_{x \in \mc{X}} A_x-\lambda I_{\mc{H}}\right)_+ \geq \sum_{x \in \mc{X}} \tr \left(A_x-\lambda I_{\mc{H}}\right)_+.
    \end{equation}
\end{lemma}

\begin{lemma} [\hspace{-0.1pt}{\cite[Theorem~\uppercase\expandafter{\romannumeral4}.4]{Mosonyi_2014}}]
\label{lem:mons}
    For any $\rho \in \mc{S}(\mc{H})$, $\sigma \in \mc{P}(\mc{H})$, $a \in \mathbb{R}$ and $t>0$, we have
    \begin{equation}
        \begin{split}
         &\lim_{n \rightarrow \infty} \frac{1}{n}\log  \tr \rho^{\ox n}\{\rho^{\ox n}>t2^{na}\sigma^{\ox n} \} \\
         =&\lim_{n \rightarrow \infty} \frac{1}{n}\log  \tr (\rho^{\ox n}-t2^{na}\sigma^{\ox n})_+ \\
         =&\inf_{s \geq 0} \left\{ s(\widetilde{D}_{1+s}(\rho\|\sigma)-a)\right\}.
        \end{split}
    \end{equation}
\end{lemma}

\begin{lemma}[\hspace{-0.1pt}{\cite[Lemma~22]{cheng2024jointstatechanneldecouplingoneshot}}]
\label{lem:avehar}
Let $\mc{H}_{\tilde{\A}}$, $\mc{H}_{\A'}$, and
$\mc{H}_{\tilde{\A}'}$ be copies of $\mc{H}_{\A}$. Then
\begin{align}
&\mathbb{E}_{\mathbb{U}(\A)}
 \left[
 U_{\A}\Phi_{\A\A'}U_{\A}^*
 \ox
 U_{\tilde{\A}}\Phi_{\tilde{\A}\tilde{\A}'}U_{\tilde{\A}}^*
 \right]
\nonumber\\
&=
\frac{
 I_{\A\A'\tilde{\A}\tilde{\A}'}
 +F_{\A\tilde{\A}}\ox F_{\A'\tilde{\A}'}
}{
 |\A|^2(|\A|^2-1)
}
-\frac{
 F_{\A\tilde{\A}}\ox I_{\A'\tilde{\A}'}
 +I_{\A\tilde{\A}}\ox F_{\A'\tilde{\A}'}
}{
 |\A|^3(|\A|^2-1)
},
\end{align}
where $F_{\A\tilde{\A}}=\sum_{i,j=1}^{|\A|} \ket{i}\bra{j}_{\A} \otimes \ket{j}\bra{i}_{\tilde{\A}}$ is the swap operator between systems $\A$ and $\tilde{\A}$.
\end{lemma}

\begin{lemma}
\label{lem:inver}
Let $(\Omega,\mathcal F,\mu)$ be a probability space, let
$f:\Omega\to\mathcal P(\mathcal H)$ be an integrable
operator-valued random variable, and let
$A\in\mathcal P(\mathcal H)$. If
    $\supp(A)\subseteq\supp(f(\omega))$
for $\mu$-almost every $\omega$, then
\begin{align}
    \mathbb E_\mu\tr A\log f(\omega)
    \leq
    \tr A\log\mathbb E_\mu f(\omega).
\end{align}
\end{lemma}
\begin{proof}
    For any $\epsilon>0$, since $\supp(A)\subseteq\supp(f(\omega))$ for $\mu$-almost every $\omega$ and $\supp(A)\subseteq\supp(\mathbb{E}_\mu(f(\omega)))$, operator monotonicity of the logarithm on $\supp(f(\omega))$ followed by operator concavity of the logarithm, gives
    \begin{equation}
    \label{equ:inver} 
        \mathbb{E}_\mu  \tr A \log( f(\omega)) 
        \leq   \mathbb{E}_\mu  \tr A \log( f(\omega)+\epsilon I) 
         \leq \tr A \log \mathbb{E}_\mu( f(\omega) +\epsilon I )      
    \end{equation}
    By letting $\epsilon \rightarrow 0$ in Eq.~(\ref{equ:inver}), we have
    \begin{equation}
  \mathbb{E}_\mu  \tr A \log( f(\omega))  \leq \lim_{\epsilon \rightarrow 0}  \tr A \log \mathbb{E}_\mu( f(\omega) +\epsilon I )  = \tr A \log \mathbb{E}_\mu f(\omega) .
    \end{equation}
\end{proof}

\begin{lemma}[$1$-LOCC entanglement generation to entanglement generation {\cite[Lemma 14.6]{khatri2024principlesquantumcommunicationtheory}}]
    \label{lem:wilde1}
    Given a $(d,\epsilon)$-achievable one-way LOCC-assisted entanglement generation protocol for a channel $\mscr{N}_{\A \rightarrow \B}$, there exists a $(d,\epsilon)$-achievable  entanglement generation protocol for $\mscr{N}_{\A \rightarrow \B}$.
\end{lemma}

\begin{lemma}[Entanglement generation to entanglement transmission {\cite[Lemma 14.7]{khatri2024principlesquantumcommunicationtheory}}]
    \label{lem:wilde2}
    Given a $(d,\epsilon)$-achievable entanglement generation protocol for a channel $\mscr{N}_{\A \rightarrow \B}$, there exists a $(d,2\epsilon)$-achievable  entanglement transmission protocol for $\mscr{N}_{\A \rightarrow \B}$.
\end{lemma}

\begin{lemma}[Entanglement transmission to strong subspace transmission]
\label{lem:wilde3}
Given a $(d,\epsilon)$-achievable entanglement transmission protocol for a
channel $\mscr{N}_{\A \rightarrow \B}$ and $\eta\in(0,1)$, there exists a
strong subspace transmission protocol of dimension at least
$\lceil(1-\eta)d\rceil$ and error at most $\epsilon/\sqrt{\eta}$.
\end{lemma}

\begin{proof}
Let $(d,\mscr{E}_{\M \rightarrow \A},\mscr{D}_{\B \rightarrow \bar{\M}})$
be a $(d,\epsilon)$-achievable entanglement transmission protocol with $\bar{\M}\cong\M$. 
Let $\{K_j\}_j$ be Kraus
operators of the effective channel
$
\mscr{D}_{\B \rightarrow \bar{\M}}
\circ\mscr{N}_{\A \rightarrow \B}
\circ\mscr{E}_{\M \rightarrow \A}$.
Define the error operator
\[
\Delta
:=
\sum_j
\left(K_j-\frac{\tr K_j}{d}I_{\M}\right)^\dagger
\left(K_j-\frac{\tr K_j}{d}I_{\M}\right).
\]
Since $\sum_jK_j^\dagger K_j=I_{\M}$, the Kraus formula for
entanglement fidelity gives
\[
\frac{1}{d}\tr\Delta
=
1-\frac{1}{d^2}\sum_j|\tr K_j|^2
=
\bigl(
p_{\rm{err}}^{\rm{ET}}(\mscr{E},\mscr{D};\mscr{N})
\bigr)^2
\leq\epsilon^2.
\]

Let $\M_\eta\subseteq\M$ be the spectral subspace of $\Delta$
corresponding to eigenvalues not exceeding $\epsilon^2/\eta$.
If $\epsilon=0$, then $\Delta=0$. If $\epsilon>0$, the preceding
bound shows that fewer than $\eta d$ eigenvalues of $\Delta$ can be
larger than $\epsilon^2/\eta$. Consequently,
\[
|\M_\eta|
\geq
\lceil(1-\eta)d\rceil,
\qquad
P_{\M_\eta}\Delta P_{\M_\eta}
\leq
\frac{\epsilon^2}{\eta}P_{\M_\eta}.
\]

Let $\bar{\M}_\eta\subseteq\bar{\M}$ be the corresponding output
subspace. Restrict the original encoder to $\M_\eta$, and define the
new decoder from $\B$ to $\bar{\M}_\eta$ by
\[
\begin{split}
\mscr{D}^{(\eta)}_{\B\to\bar{\M}_\eta}(X)
:={}&
P_{\bar{\M}_\eta}
\mscr{D}_{\B \rightarrow \bar{\M}}(X)
P_{\bar{\M}_\eta}
\\
&+
\tr\!\left[
(I_{\bar{\M}}-P_{\bar{\M}_\eta})
\mscr{D}_{\B \rightarrow \bar{\M}}(X)
\right]
\frac{P_{\bar{\M}_\eta}}{|\M_\eta|}.
\end{split}
\]

Let $\psi$ be any pure input state on $\M_\eta$ and an arbitrary
reference system, and let $\psi_{\M_\eta}$ be its marginal on
$\M_\eta$. Since the target state is supported on
$\bar{\M}_\eta$, the definition of $\mscr{D}^{(\eta)}_{\B\to\bar{\M}_\eta}$ and data-processing inequality of fidelity give
\begin{align*}
&F^2\!\left(
\bigl(
\id\otimes
(\mscr{D}^{(\eta)}_{\B\to\bar{\M}_\eta}
\circ\mscr{N}_{\A \rightarrow \B}
\circ\mscr{E}_{\M_\eta \rightarrow \A})
\bigr)(\psi),
\psi
\right)
\\
&\quad\geq
F^2\!\left(
\bigl(
\id\otimes
(\mscr{D}_{\B \rightarrow \bar{\M}}
\circ\mscr{N}_{\A \rightarrow \B}
\circ\mscr{E}_{\M_\eta \rightarrow \A})
\bigr)(\psi),
\psi
\right)
\\
&\quad=
\sum_j
\left|
\tr(K_j\psi_{\M_\eta})
\right|^2
\\
&\quad\overset{(\star)}{=}
1-\tr(\Delta\psi_{\M_\eta})
+
\sum_j
\left|
\tr\!\left[
\left(
K_j-\frac{\tr K_j}{d}I_{\M}
\right)
\psi_{\M_\eta}
\right]
\right|^2
\\
&\quad\geq
1-\tr(\Delta\psi_{\M_\eta})
\\
&\quad\geq
1-\frac{\epsilon^2}{\eta},
\end{align*}
where $(\star)$ follows by expanding
\[
\sum_j
\tr(K_j^\dagger K_j\psi_{\M_\eta})
=
1.
\]
Hence the purified-distance error of the new strong subspace
transmission protocol is at most $\epsilon/\sqrt{\eta}$.
\end{proof}

\begin{lemma}
\label{lem:reldt}
	For any quantum states $\rho$ and $\sigma$, we have
	\begin{equation}
		D(\rho\Vert\sigma) \geq \tr(\rho - 9\sigma)_{+}. \label{eq:lemma1}
	\end{equation}
\end{lemma}

\begin{proof}
	Set $Q := \{\rho > 9\sigma\}$. Denote $p = \tr(\rho Q)$ and $q = \tr(\sigma Q)$, which are the probabilities of obtaining the outcome associated with $Q$ when a projective measurement $\{Q, I-Q\}$ is applied to $\rho$ and $\sigma$, respectively. Then, it is easy to see that
	\begin{equation}
		Q \rho Q \geq 9 Q \sigma Q,
	\end{equation}
	which gives
	\begin{equation}
		p \geq 9q. \label{eq:pgeq9q}
	\end{equation}
	
	Now by the monotonicity of the quantum relative entropy under quantum measurements, we have
	\begin{align}
		D(\rho\|\sigma) &\geq D\big( (p, 1-p)\|(q, 1-q) \big) \\
		&\geq  p\log\frac{p}{q}+(1-p)\log \frac{1-p}{1-q} \\
		&\geq 3p + (1-p)\log(1-p), 
	\end{align}
where the last line is from Eq.~\eqref{eq:pgeq9q}. Noticing that
$2x+(1-x)\log(1-x) \geq 0$ for $x \in [0,1]$,
we obtain
	\begin{align}
	D(\rho\| \sigma) \geq p  \geq  \tr(\rho - 9\sigma)_{+}.
	\end{align}
\end{proof}


\bibliography{ref.bib}

\end{document}